\documentclass[10pt]{article}
\usepackage[latin1]{inputenc}
\usepackage[T1]{fontenc}
\usepackage{graphicx}
\usepackage{bm}
\usepackage{amsmath}
\usepackage{dsfont}
\usepackage{lipsum,multicol}
\usepackage{amssymb}
\usepackage{graphicx}
\usepackage{pifont}
\usepackage{nopageno}
\usepackage{algpseudocode}
\usepackage{amsfonts}
\usepackage{appendix}
\usepackage{amscd,amsthm}
\usepackage{mathtools}  
\usepackage{tabulary}
\usepackage{booktabs}
\usepackage{algorithm}
\usepackage{hyperref} 
\usepackage{multirow}
\usepackage[font=scriptsize,labelfont=bf,margin=0.1in]{caption}
\newtheorem{Satz}{Theorem}[section]

\newtheorem{Prop}[Satz]{Proposition}

\newtheorem{Conj}[Satz]{Conjecture}

\usepackage[margin=.6in]{geometry}
\usepackage{xr}
\title{A rigorous and efficient asymptotic test for power-law cross-correlation}
\author{Duncan A.J.~Blythe}
\date{\today}
\begin{document}
\maketitle
\abstract{ Podobnik and Stanley recently proposed a novel framework \cite{DCCA}, Detrended Cross-Correlation Analysis, for the analysis of power-law cross-correlation between two time-series, a phenomenon which 
occurs widely in physical, geophysical, financial and numerous additional applications. While highly promising in these important application 
domains, to date no rigorous or efficient statistical test has been proposed which uses the
information provided by DCCA across time-scales for the presence of this power-law cross-correlation.
In this paper we fill this gap by proposing a method based on DCCA for testing the hypothesis of power-law cross-correlation; the method synthesizes the information generated by DCCA across time-scales and returns conservative but practically relevant $p$-values for the null hypothesis of zero correlation, which may be efficiently calculated in software. 
Thus our proposals generate confidence estimates for a DCCA analysis in a fully probabilistic fashion.}

\section{Introduction}

The presence of power law autocorrelations in empirical time-series has long been noted and studied in a broad range of physical applications \cite{mandelbrot1982fractal,leland1994self,bak1987self}. In particular, the widespread observation  of power-law autocorrelation in the time-domain has stimulated the development of numerous methods for the estimation
of the exponent of the power law, including Detrended Fluctuation Analysis \cite{PengDFA}. Moreover since the individual time-series are often univariate components of a multi dimensional and complex physical system, for example in neuroscience (e.g. EEG) and finance (e.g. stock index time-series), natural objects of study are the \emph{interactions} between power-law correlated time-series. In particular, under certain assumptions, if these time-series arise from a pair of components
between which interactions exist, then the time-series are power-law \emph{cross}-correlated: that is the time-lagged cross-correlation function between the pair of time-series takes the form of a power-law asymptotically. Accordingly Podobnik and Stanley's
Detrended Cross-Correlation Analysis (DCCA) \cite{DCCA}  constitutes an important tool in analyzing these interactions, by extending DFA to the analysis of cross-correlations across scales between two time-series. Since its recent development DCCA has been applied in finance \cite{DCCA_world_stock,wang2011quantifying}, geophysics \cite{DCCA_physics, shadkhoo2009multifractal, hajian2010multifractal,shang2009chaotic} and socio-geographical data \cite{zebende2011study}. Concomitantly, several important methodological extensions of DCCA have been proposed, including a multi-fractal extension \cite{zhou2008multifractal} 
and heuristics for using DCCA to test for the presence of power-law cross-correlation \cite{podobnik2011statistical,zebende2011dcca}.
Despite these numerous innovations, and as we shall see, it remains unclear exactly how to use the information yielded by DCCA to perform 
a \emph{statistical test} for power-law cross-correlation which conforms to the highest levels of scientific rigour and which may be executed efficiently without recourse to large scale computation;
as we shall see below in Section~\ref{sec:review}, using DCCA without modification cannot be used to test statistically for the presence long-range cross-correlation and the only contribution of which we are aware in the literature which does purport to yield a statistical test using DCCA
does not satisfy the dual desiderata of rigour and efficiency. This paper proposes a statistical test based on DCCA which allows the data analyst to compute a $p$-value for the null-hypothesis of independence, which may be executed
within seconds in a modest computing environment and which may be shown to be correct within a flexible semi-parametric class of long-range dependent time-series which are possibly contaminated by non-stationary polynomial trends of an arbitrary degree and thus conforms to these standards of both rigour and efficiency.

\subsection{Review of DCCA and the problem setting}
\label{sec:review}
We assume in the following that we are given two time series $Y_1(t)$ and $Y_2(t)$ at least one of which is power-law autocorrelated. I.e., $Y_1(t)$ and $Y_2(t)$ are subject to \emph{Hurst exponents}: $H,G \geq 1/2$ s.t. $H$ or $G>1/2$  and as $s \rightarrow \infty$:

\begin{equation}
 \mathbb{E}(Y_1(t)Y_1(t+s)) \sim \frac{C}{s^{2H-2}} \text{\hspace{0.2in}or\hspace{0.2in}} \mathbb{E}(Y_2(t)Y_2(t+s)) \sim \frac{C'}{s^{2G-2}} \label{eq:LRTC}
 \end{equation}
 \newline
 Moreover, in this context of long-range dependence we are interested in testing for the presence of power-law \emph{cross-correlation}, which we will also refer to as LRCC for short, i.e., whether for some $\beta>0$ and $A>0$:  
 \begin{equation}
 \mathbb{E}(Y_1(t)Y_2(t+s)) \sim \frac{A}{s^{\beta}} \label{eq:LRCC}
 \end{equation}
 \newline
 Thus, the phenomenon in which we are interested may be completely defined in terms of the cross-covariance and autocovariance functions of the time-series of interest (equivalently power-spectra). However, although the topic of interest may be defined in terms of covariances, in practice the empirical time-lagged covariances and cross-covariances are unreliable as guides to the analysis of power-law properties in the limit $s \rightarrow \infty$. This is firstly because analyzing the behaviour of time-lagged correlations
 as $s \rightarrow \infty$ implies using the information in the lowest-frequencies of the given signals, which are often contaminated by deterministic trends, for example by recording artifacts in biomedicine:
 low pass filtering to clean these artifacts would delete valuable information contained in the low-frequencies regarding the presence of long-range correlation and cross-correlation.
 Secondly, under the assumption of long-range dependence, certain desirable statistical properties do not hold for the sample lagged correlations, which would otherwise facilitate estimation\footnote{These include, for example, failure to converge to Gaussianity as the number of time-samples grows and minimax convergence rates in estimation} \cite{bardet2000testing}.

 For these reasons numerous novel methodologies have already been proposed for estimating the parameters of and testing for the long-range \emph{autocorrelation} of Equation~\eqref{eq:LRTC}, which circumvent 
 the cited difficulties which apply when using the empirical autocorrelation function. These include Detrended Fluctuation Analysis (DFA) \cite{PengDFA}, which may be used to estimate the \emph{Hurst} parameter $H$ of a given univariate time-series, $Y(t)$. 
 DFA can be proven to possess the desirable statistical properties lacking in a autocovariance analysis \cite{DFA_asymp} and implements
 detrending so that deterministic non-stationarities, of potentially arbitrary polynomial order \cite{kantelhardt2001detecting}, have no influence on estimation of $H$. More precisely, DFA involves first forming the aggregate sum of the empirical 
 time-series $Y(t)$:
 \begin{equation}
X(t) = \sum_{i=1}^t Y(i)
 \end{equation}
 
(From now on whenever we refer to $X_j(t)$ we mean the time-series obtained from $Y_j(t)$ by way of this operation.)
Analysis of the fluctuations in $X(t)$ may then be performed by measuring the variance of $X(t)$ in windows of varying size $n$ \emph{after} detrending, i.e., $X(t)$ is split into windows
of length $n$, $X_n^{(1)},\dots,X_n^{(j)},\dots,X_n^{([N/n])}$ and the average variance after detrending the data in these windows is formed; i.e. let $\mathbb{P}_d$ be the operator which generates the mean-squares estimate of the polynomial fit of degree $d$, then the DFA coefficients or detrended variances of degree $d$ are:

\small
\begin{equation}
F_{DFA}^2(n) = \frac{1}{n} \sum_{j}\left(X_n^{(j)}-\mathbb{P}_d\left(X_n^{(j)}\right)\right)^\top\left(X_n^{(j)}-\mathbb{P}_d\left(X_n^{(j)}\right)\right)
\end{equation}
 \normalsize
Crucially, it is possible to show that in the limit of data the slope of $\text{log}(F^2_{DFA}(n))$ against $\text{log}(n)$ converges to $H$ \cite{DFA_asymp}. Thus $Y(t)$ is power-law correlated if and only if
the estimate of $H$, $\widehat{H}$, converges to a number greater than $0.5$ in the limit of data. 

In precise analogy, Podobnik and Stanley propose Detrended Cross-Correlation Analysis (DCCA), an extension of DFA to two time-series, by considering:
 \small
\begin{multline}
F_{DCCA}^2(n) = \\ \frac{1}{n} \sum_{j}\left((X_1)_n^{(j)}-\mathbb{P}_d\left((X_1)_n^{(j)}\right)\right)^\top\left((X_2)_n^{(j)}-\mathbb{P}_d\left((X_2)_n^{(j)}\right)\right)
\end{multline} 
 \normalsize
 
 Thus DCCA generalizes DFA in the sense that if $X_1=X_2$ then $F_{DCCA}^2(n) = F_{DFA}^2(n)$. 
 To simplify the fact that we will need to consider the DFA coefficients of $Y_1$ and $Y_2$ simultaneously, we will also refer to the DCCA coefficients as:
 \begin{eqnarray}
 F_{X_1,X_2}^2(n) &:=& F_{DCCA}^2(n)
 \end{eqnarray}

Given these definitions, it is possible to show (Proposition~\ref{prop:expectation} of the Appendix) that:
\begin{enumerate}
\item If $Y_1(t)$ and $Y_2(t)$ are independent then  \newline $\mathbb{E}(F_{DCCA}^2(n)) = 0$
\item If $Y_1(t)$ and $Y_2(t)$ are positively correlated, then, on average, and for large window sizes  $\text{log}(F_{DCCA}^2(n))$ is linear against $\text{log}(n)$ 
\item If $Y_1(t)$ and $Y_2(t)$ are negatively correlated, then, on average, and for large window sizes  $\text{log}(-F_{DCCA}^2(n))$ is linear against $\text{log}(n)$ 
\end{enumerate}

This implies that for sufficiently large datasets we should be able to distinguish between those which are power-law cross-correlated and those which are not power-law cross-correlated by checking 
which of these cases applies; this is also the approach taken in the paper in which DCCA was originally proposed \cite{DCCA}. Given that we may conclude dependence, we may then further test for power-law correlation by checking whether
the absolute value of the slope of  $\text{log}(F_{DCCA}^2(n))$ against $\text{log}(n)$ is greater than $1/2$.  See Figure~\ref{fig:easy_examples} for two examples
of resp.~independence and dependence for which this method works.

\begin{figure}
\begin{center}
$\begin{array}{c c c c}
\includegraphics[width=40mm,clip=true,trim= 0mm 0mm 8mm 0mm]{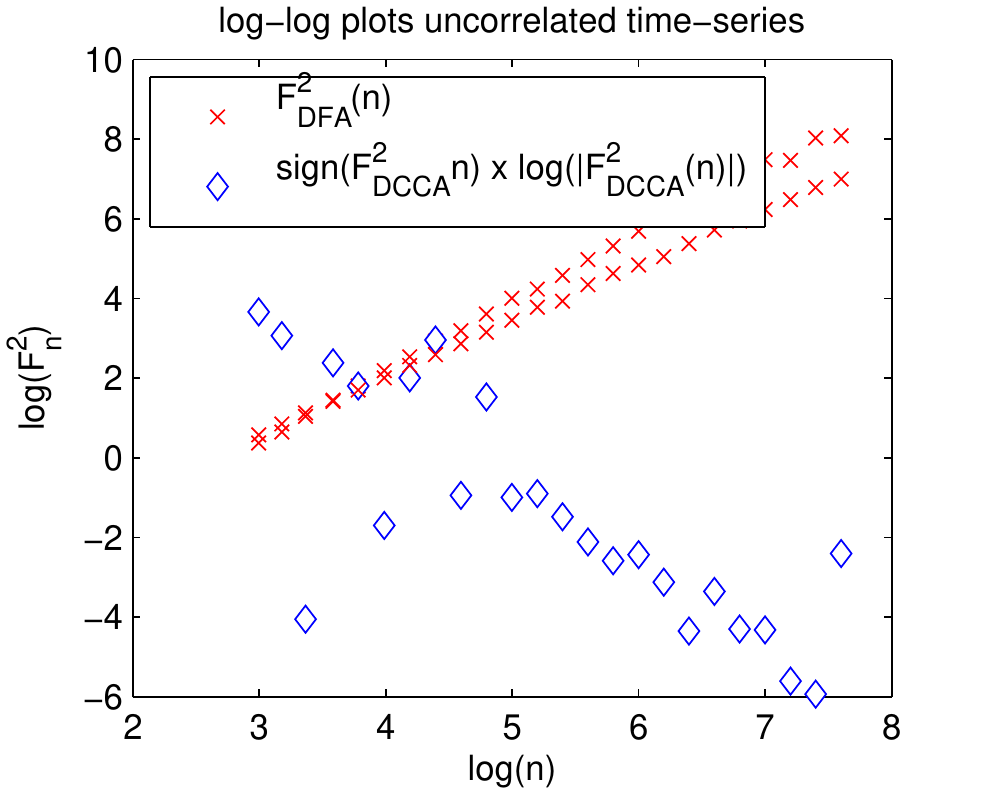} & \includegraphics[width=40mm,clip=true,trim=  0mm 0mm 8mm 0mm]{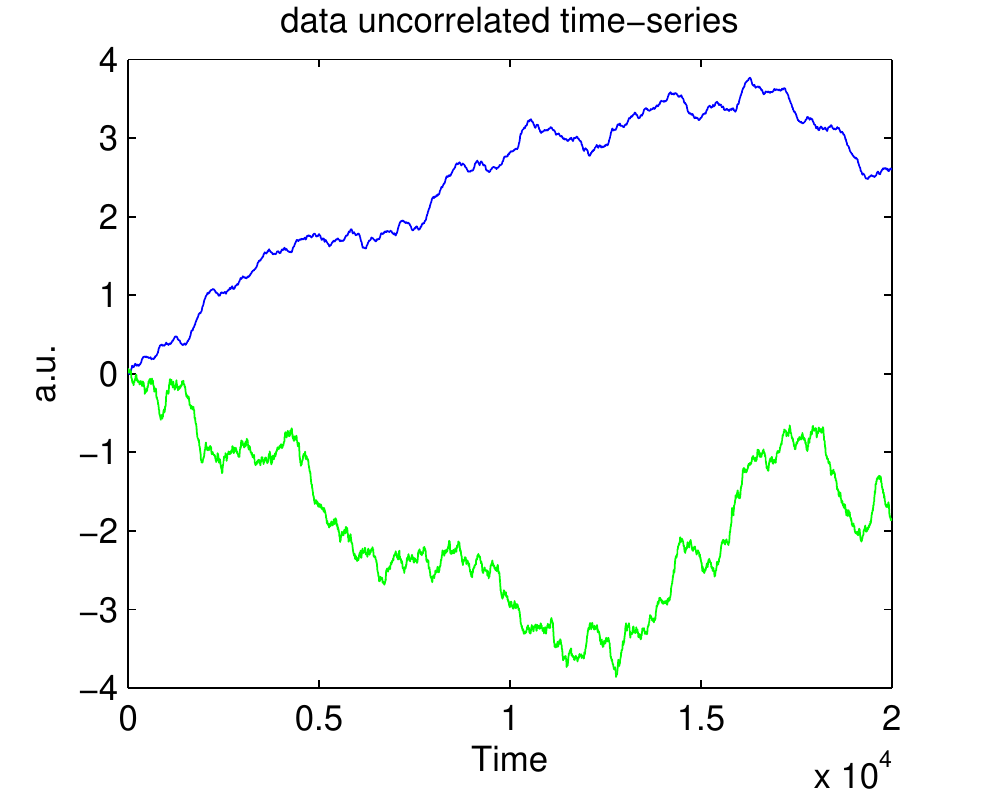} \\
\includegraphics[width=40mm,clip=true,trim=  0mm 0mm 8mm 0mm]{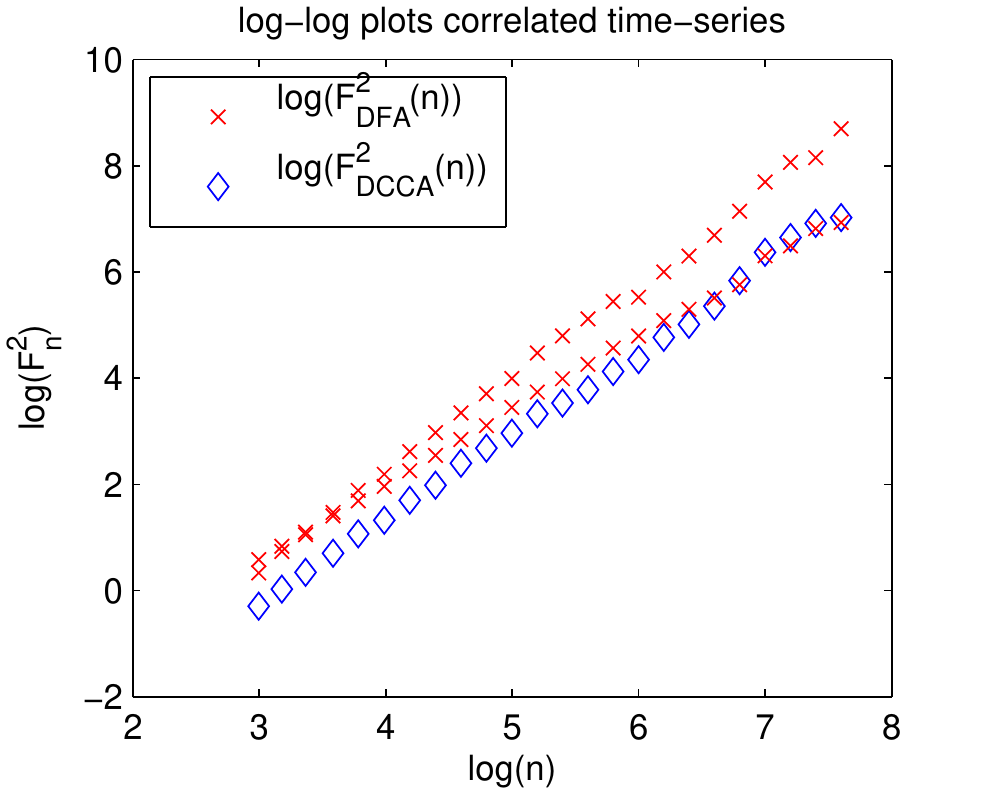} & \includegraphics[width=40mm,clip=true,trim=  0mm 0mm 8mm 0mm]{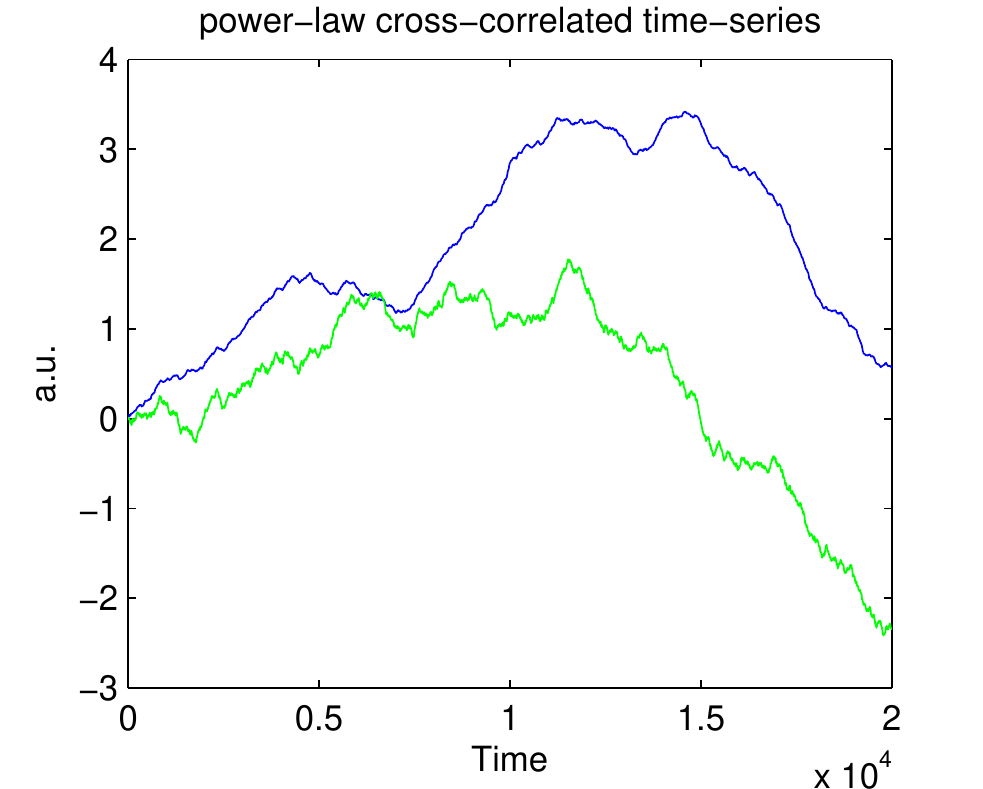} 
\end{array}$
\end{center}
\caption{Illustration of DCCA for power-law cross-correlated and uncorrelated time-series; the figure displays in the left-column the log-DCCA coefficients and log-DFA coefficients plotted against the log-window size for resp. uncorrelated and correlated fractional Gaussian Noises.
In each case the fractional Gaussian noises are of length 20000 data points, and are subject to Hurst exponents of 0.7 and 0.9 resp.. The aggregated data  ($X_1,X_2$) are plotted in the right hand column in each case.
In the first case the time-series are uncorrelated ($\rho = 0$) and thus the DCCA coefficients oscillate around 0 and display no log-linear behaviour (upper row).
In the second case the time-series are correlated ($\rho = 0.4$) and thus the DCCA coefficients display log-linear behaviour (bottom row left).
\label{fig:easy_examples}}
\end{figure}

In each case we generate two fractional Gaussian noises\footnote{We use fractional Gaussian noise for illustration because it possesses certain properties of special significance which are possessed asymptotically by all time-series within the model of Equations~\eqref{eq:LRTC} and~\eqref{eq:LRCC}; these properties are discussed later in the paper.}, time-series which conform to the model of Equations~\eqref{eq:LRTC} and~\eqref{eq:LRCC}, with Hurst parameters
$0.9$ and $0.7$ respectively and length $20000$ time-points; we choose $n$ to range between $20$ and $2000$. The presence of long-range dependence may be concluded by looking at the slope of $\text{log}(F_{DFA}^2(n))$ against $\text{log}(n)$ which is greater than $1/2$ in each case (in red). 

In the first case the time-series are chosen to be independent; here we see that $F^2_{DCCA}(n)$ fluctuates around zero and that no linear dependence of $\text{log}(F_{DCCA}^2(n))$ against $\text{log}(n)$ is present.
To visualize this we plot $\text{sign}(|F_{DCCA}^2(n))) \times \text{log}(|F_{DCCA}^2(n)|)$ which brings $F^2_{DCCA}(n)$ onto the log scale but allows for inspection of fluctuation around zero.
In the second case, the time-series are correlated with correlation equal to $0.4$: we see that $F^2_{DCCA}(n)$ does not fluctuate around zero and that $\text{log}(F_{DFA}^2(n))$ is linear against $\text{log}(n)$; the slope of the straight line fit is greater than $1/2$, suggesting the presence of long-range cross-correlation.

\subsubsection{Problems with interpreting the DCCA log-log plot}

Thus, testing for power-law cross-correlation under this scheme is a \emph{two step} procedure. The second step is formally equivalent to the use of DFA to test for long-range autocorrelation. The \emph{first} step however has no analogy for DFA;  indeed it is this first step which is problematic. See Figure~\ref{fig:type_I_errors} for examples of log-log plots for independent fractional Gaussian noises when there is less data available ($N=10000$) and the window sizes $n$ run between $10^1$ and $10^3$.
The left hand panel displays $\text{sign}(|F_{DCCA}^2(n))) \times \text{log}(|F_{DCCA}^2(n)|)$ plotted against $\text{n}$. A straight line is clearly visible, even though the underlying time-series are independent. 
To demonstrate that such cases are not pathological but occur frequently, we plot, in the right hand panel, the significance level for a $t$-test (two sided test on Pearson's correlation coefficient at the $p=0.05$ level) for zero correlation as a histogram for 100 simulated fractional Gaussian noises. In 35 cases
the null-hypothesis is rejected at the 0.05 significance level: this demonstrates that the statistics for the appearance of a linear relationship in the DCCA log-log plot are quite different to the statistics for standard regressive testing.
Thus the presence of a linear relationship between the \emph{empirical} $\text{log}(F^2_{X_1,X_2}(n))$ and $\text{log}(n)$ is unreliable as a guide to interaction; in particular, for smaller data sets, where fewer time-scales are available, the presence of an approximately linear relationship may occur spuriously 
due to the correlations between $F^2_{X_1,X_2}(n_1)$ and $F^2_{X_1,X_2}(n_2)$ for $n_1$ and $n_2$ within a few degrees of magnitude and due to the increasing variance of the $F_{DCCA}^2(n)$ coefficients in $n$. 
The extra information w.r.t. dependence, omitted by considering the slope, in the case displayed in  the left-hand panel of Figure~\ref{fig:type_I_errors}, is contained in the \emph{height} of  the estimated
straight line fit in relation to the DFA straight line fits; compare with the left hand panel of Figure~\ref{fig:easy_examples}. Thus, the \emph{closer} the DCCA coefficient to the DFA coefficients, the more likely it is that correlation exists between the time-series.

\begin{figure}
\begin{center}
$\begin{array}{c c}
\includegraphics[width=40mm,clip=true,trim= 0mm 0mm 0mm 0mm]{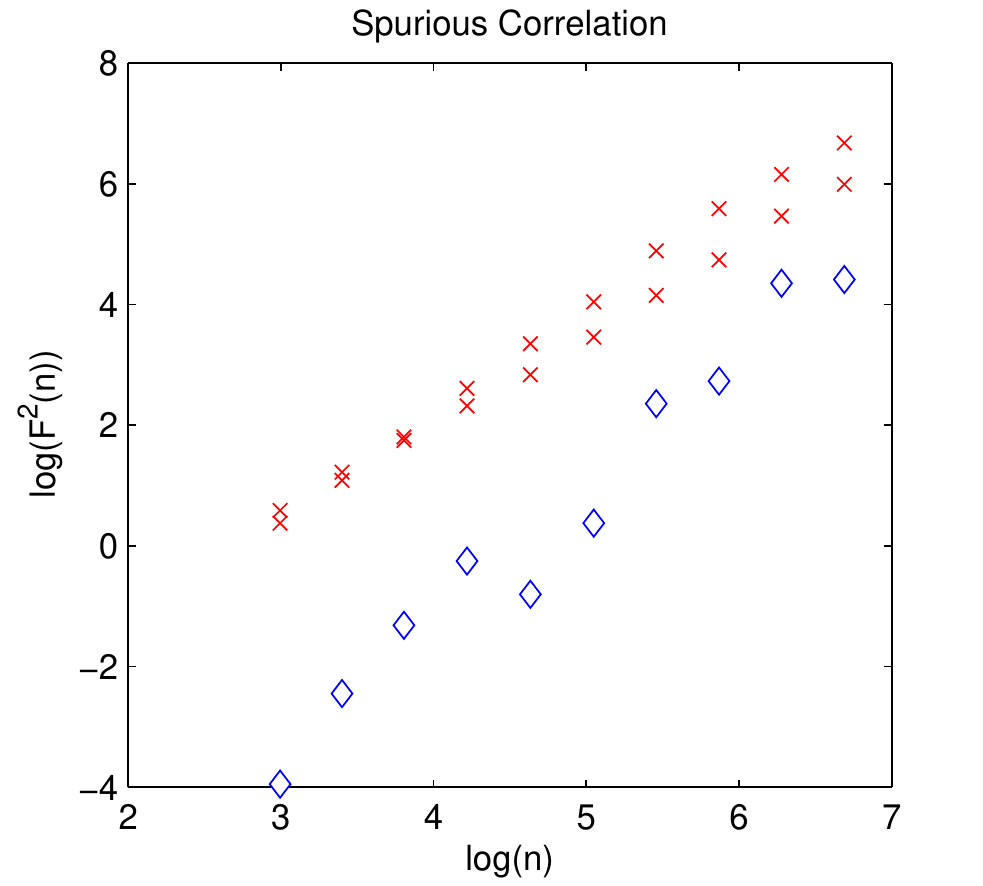}
& \includegraphics[width=40mm,clip=true,trim= 0mm 0mm 0mm 0mm]{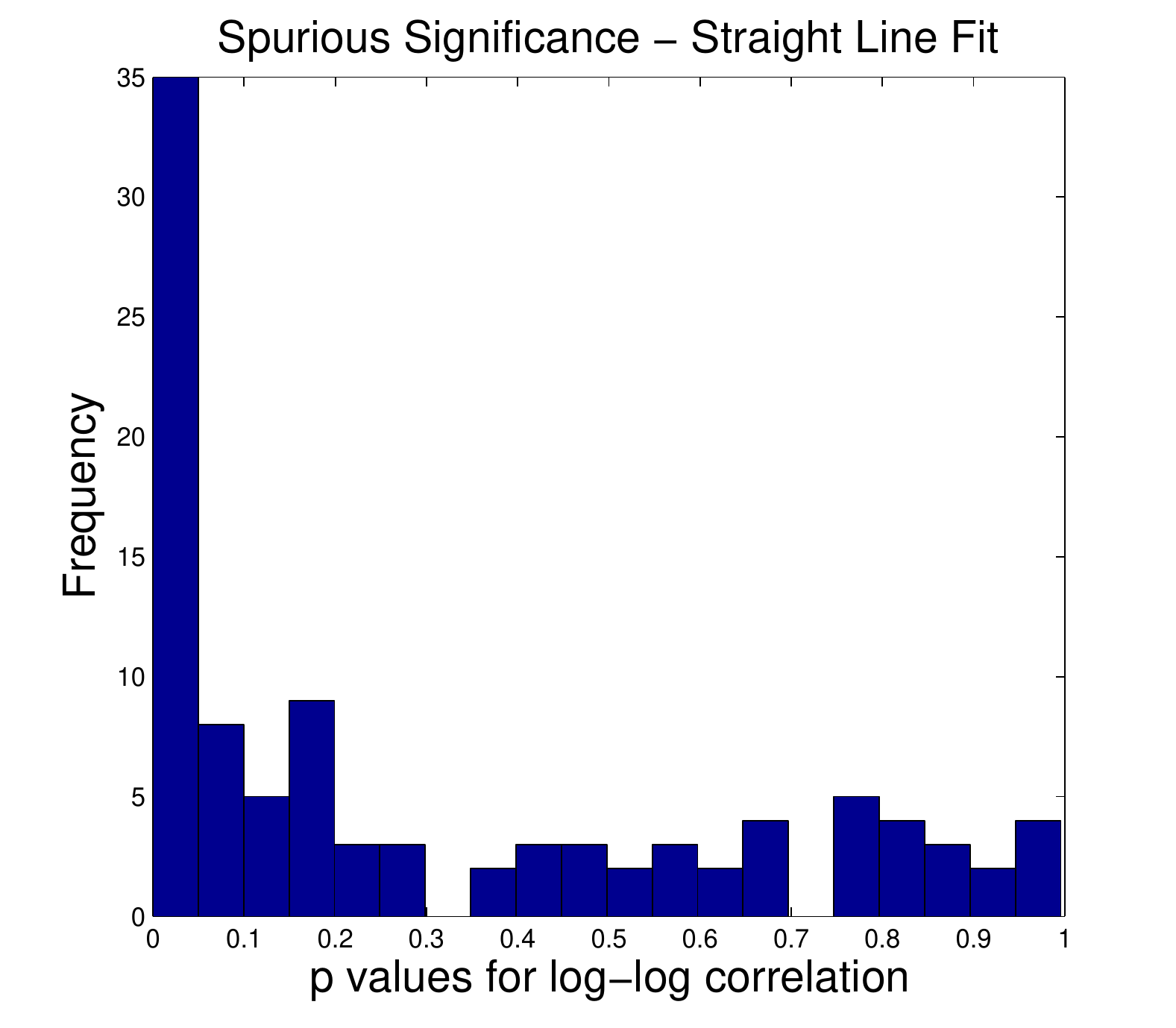}\end{array}$ 
\end{center}
\caption{The problem of spurious linear relationships in the DCCA log-log plot; the figure displays in the left hand panel a case in which the window lengths $n_i$ lie between $10^1$ and $10^3$ ($H=0.7,G=0.8$) and log-linear behaviour is observed in the DCCA coefficients even though 
the underlying time-series are independent. 
Spurious linear-behaviour is common, as demonstrated by the right hand panel, where the frequency of rejection of the hypothesis of no correlation between the log-DCCA coefficients and the log-scale is displayed.
In approximately one-third of cases a linear relationship is detected on the basis of a $t$-test even though no underlying correlations are present in the time-series.
\label{fig:type_I_errors}}
\end{figure}

\subsubsection{The DCCA correlation coefficients and statistical testing}
\label{sec:DCCA_testing}
In accordance with the observation that the \emph{magnitude} of the $F_{DCCA}^2$ coefficients carry information relating to the strength of dependencies, Podobnik et al. study, in a later paper ~\cite{podobnik2011statistical}, the \emph{DCCA correlation coefficient}, initially introduced by \cite{zebende2011study}:
\begin{equation}
\rho_{DCCA}(n,X_1,X_2) = \frac{F^2_{X_1,X_2}(n)}{\sqrt{F^2_{X_1,X_1}(n)}\sqrt{F^2_{X_2,X_2}(n)}} \label{eq:rho}
\end{equation}

The coefficient is analogous to Pearson's correlation coefficient, in that the terms on the denominator are the detrended \emph{variances} which are utilized in DFA and the numerator DCCA coefficient is a detrended \emph{cross-covariance}; moreover, as for Pearson's coefficient, $\rho_{DCCA}$ may be shown to lie between $-1$ and $1$ \cite{podobnik2011statistical} and thus presents a promising starting point in testing for long-range interdependence.

However, testing with $\rho_{DCCA}$ is a very different problem to testing with a standard correlation coefficient: the distributional 
theory which applies to the standard coefficient does not apply here on account of the detrending.  In accordance with this fact Podobnik et al.~propose to construct a test for power-law cross-correlation by estimating the quantiles of $\rho_{DCCA}(n,X_1,X_2)$ using a FARIMA assumption and comprehensive computer simulation of time-series of the same length as the empirical time-series. In doing so they provide a guide to interpreting the significance of observing a particular value
of $\rho_{DCCA}(n,X_1,X_2)$ on a given data set by providing insight into the range of $\rho_{DCCA}$ values which may arise under the null hypothesis of independence due to small sample fluctuations.
See Figure~\ref{fig:rho_test}, which displays empirical DCCA correlation coefficients and the critical level $\rho_c$ for each coefficient for a power-law correlated time-series ($\rho = 0.3$- fractional Gaussian noise) and an uncorrelated case ($\rho = 0$): we see that for the uncorrelated case, all DCCA correlation coefficients lie between the 0.05 level critical boundaries, whereas, in the correlated case, all DCCA correlation coefficients lie above the upper boundary (denoting positive correlation).

This proposal does not consider, however the complications induced by the fact that one obtains \emph{several} $\rho_{DCCA}$ coefficients for every data set considered, depending on the number of scales considered.
The first difficulty this generates is that since we are faced with multiple DCCA correlation coefficients, simply checking whether any of the coefficients lie above the critical level is susceptible to multiple
testing errors; when many coefficients are measured, random fluctuations under the null hypothesis will more probably drive at least one empirical coefficient into the critical region. The standard technique for dealing with this issue is to perform a simple multiple testing-correction; one would then report long-range dependence when at least one coefficient is measured as significant after this correction. 

Testing the individual scales separately and then correcting, however, is at odds with the fact that long-range dependence
is a \emph{broadband} phenomenon; thus information relating to long-range temporal dependence should be visible \emph{across} time-scales, not simply in one scale displaying highly significant results. 
Moreover, since the effective sample size generating the coefficients at small scales is greater than at large scales, the coefficients at small scales will display significant readouts even under spurious and weak short-range dependence but long-range independence. Using a multiple testing-correction in such cases
will results in a fallacious rejection of the null-hypothesis of long-range independence. Thus since the coefficients at small scales are highly sensitive to short range cross-correlation properties, a multiple testing correction won't succeed  in controlling the Type I error rate in cases in which the short range properties of the true distribution differ from its long-range properties.
See, the right hand panel Figure~\ref{fig:rho_test} for a case in which 
the multiple testing correction leads us astray. 

\begin{figure}
\begin{center}
$\begin{array}{c c}
\includegraphics[width=40mm,clip=true,trim= 0mm 0mm 0mm 0mm]{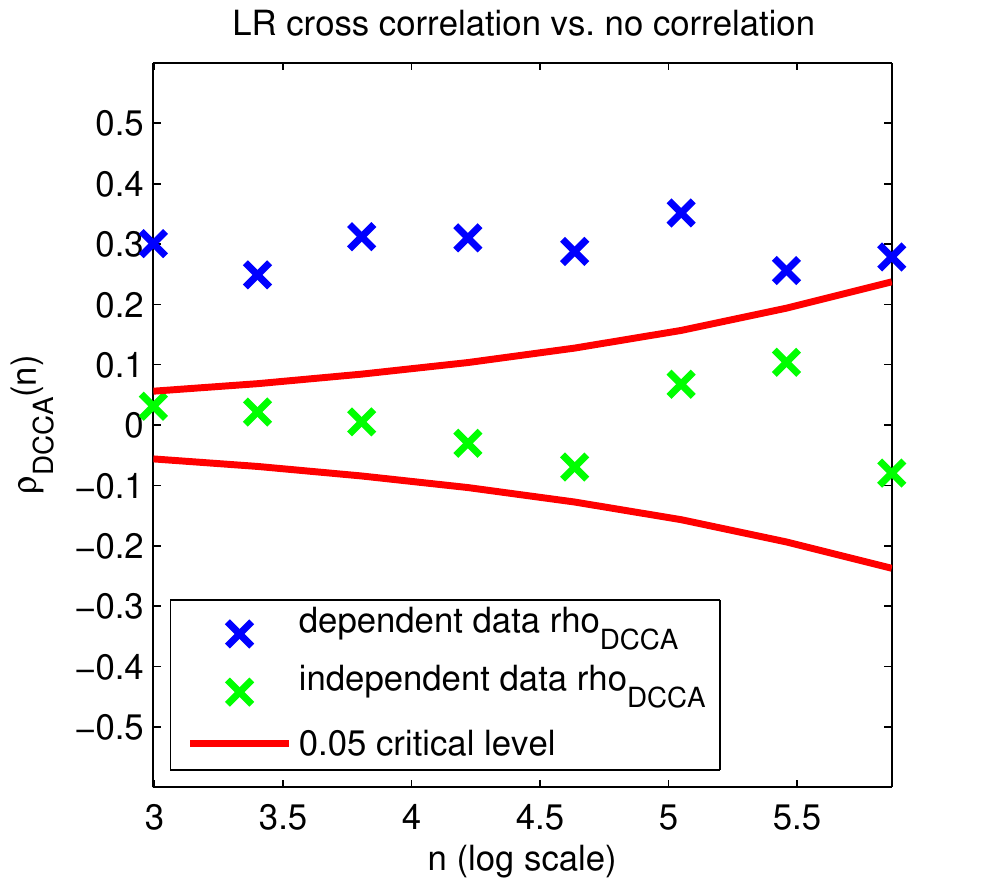} & 
\includegraphics[width=40mm,clip=true,trim= 0mm 0mm 0mm 0mm]{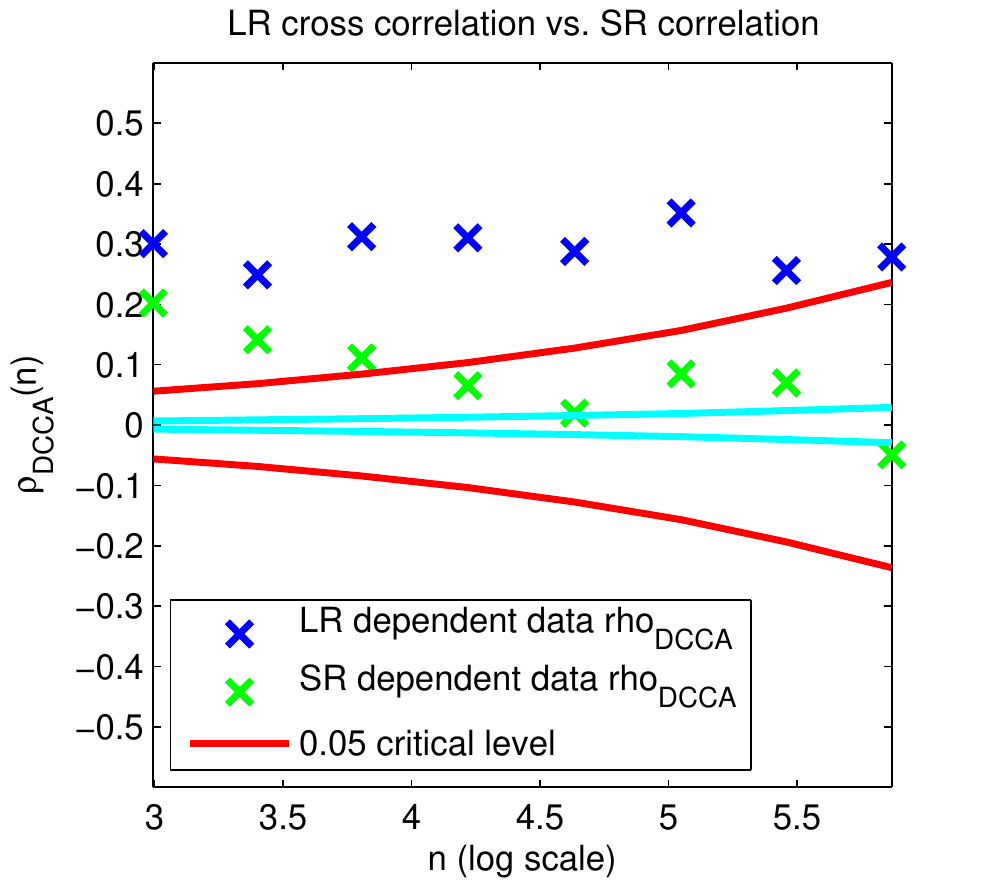}
\end{array}$
\end{center}
\caption{The effect of Spurious short-range correlations on the DCCA correlation coefficients; the figure illustrates that spurious rejection of the null may occur in the presence of short-range correlation. In particular in the left-hand panel we see that in the absence of 
correlations, all coefficients lie below the 0.05 quantile (red) of the $\rho_{DCCA}$ coefficients and when the time-series are correlated on a long-range basis all coefficients lie above the critical line (bivariate fractional Gaussian noise, $H=G=0.9$).
However when each time-series is LRTD but short range correlations exist between the time-series, in the case of \emph{no} long-range cross-correlation, in the right hand panel one sees that the  $\rho_{DCCA}$ coefficients lie above
the critical line at small scales and below at large scales.
 \label{fig:rho_test}}
\end{figure}

The blue points display the $\rho_{DCCA}$ values of two long-range dependent time-series; all points lie above the given critical threshhold. The green points, however display the $\rho_{DCCA}$ coefficients of a bivariate time-series simulated as a superposition of two independent long-range dependent fractional Gaussian noises and \emph{correlated} Gaussian white noises. Thus
the time-series are long-range \emph{independent} but short-range \emph{dependent}. Since we are interested in long-range properties, it is desirable that the testing-procedure should be made robust to these short-range correlation properties. In Figure~\ref{fig:rho_test} we see that the DCCA correlation coefficients (green) exceed the critical boundaries at small scales, due to short range dependence but not at large scales due to long-range independence.
This implies that a test based simply on performing a multiple testing-correction will lead to erroneous rejection of the null-hypothesis when the short-range component of the time-series contains correlations. 

\subsubsection{Proposed Solutions}

The alternative which we propose in this paper involves calculating 
the probability that \emph{all} DCCA correlation coefficients exceed a certain level \emph{simultaneously}. See the right hand panel of Figure~\ref{fig:rho_test} for an illustration of this procedure. Under the null hypothesis of long-range independence but with short range dependence, the DCCA correlation coefficients at small scales lie far above the significance level corrected for multiple comparisons (red). On the other hand, at higher scales, coefficients behave similarly to the coefficients of fractional Gaussian noise and thus the null hypothesis is not rejected using the proposed method based on requiring that \emph{all} coefficients lie above the upper turquoise boundary. 
Thus, critically, the procedure we propose requires using the \emph{full} distribution of the $\rho_{DCCA}$ coefficients across scales.

An immediate question which arises is: what if the opposite situation applies to that illustrated by the right hand panel of Figure~\ref{fig:rho_test}? I.e.: the pair of time-series are long-range dependent but no correlation is 
present in the short range component. This will imply that the $\rho_{DCCA}$ coefficients will point to dependence at large scales but not at small scales. Thus there may be no critical line above which all coefficients probably lie in the case of
dependence. We show in Section~\ref{sec:method}, that the information in the joint distribution of the $\rho_{DCCA}$ coefficients may be used to cater to this case while remaining robust to spurious short-range correlations. 

Once the decision has been made to use the joint distribution of the $\rho_{DCCA}$ coefficients then we need to ask: how should the distribution of the $\rho_{DCCA}$ coefficients be calculated in practice? Derivation of a closed form solution of the exact probability density function even for a parametric family of 
Gaussian process time-series presents a formidable task of analysis due to the high dimensionality of the probability space and the detrending operation which is implicit in calculation of the coefficients.
In the parametric setting, Podobik et al.~\cite{podobnik2011statistical} resort to estimation via simulated surrogate data for estimation of the distribution of a single $\rho_{DCCA}$ coefficient. On the basis of these calculations, the authors conjecture 
that the coefficients are asymptotically normal; however, since no expression is given for the parameters of this normal limit, nor an analysis of the convergence rate, simulated data is still required for its calculation.
This approach may be in principle extended to estimation of the joint distribution of the $\rho_{DCCA}$ coefficients, however, for long time-series the ansatz is computationally inefficient.  For example for estimation of 10,000 time-series, when $N=80,000$, which is approximately the 
number required in order to estimate the type I error distribution in a stable manner, we estimate on a standard desktop system that approximately $28$ hours of computation is necessary.
On the other hand, in cases when the ansatz \emph{is} efficient, empirical data is often insufficient to reliably estimate long-range properties of the data; thus cases in which the time-series considered are short are of less practical 
interest.

Quite apart from these computational difficulties, since we specify the model of interest in a semi-parametric fashion as per Equations~\eqref{eq:LRTC} and~\eqref{eq:LRCC} there is no guarantee that simulation from a parametric model will yield accurate results in approximating the true distribution of the $\rho_{DCCA}$ coefficients.
This approach may only be regarded as valid if we may relate the parametric class chosen on the one hand to behaviour across the 
semi-parametric class on the other; typically this will be effected by means of asymptotic analysis \cite{DFA_asymp,moulines2007spectral}. However, given precise asymptotic estimates across the semi-parametric class, the simulation technique
is no longer necessary and the estimates may be used directly.
We show in this paper in theory and simulation that there exists a central limit to which the $\rho_{DCCA}$ coefficients converge which is universal with respect to a broad
delineation of the semi-parametric class of interest of Equations~\eqref{eq:LRTC} and~\eqref{eq:LRCC}. In addition we show that this limit yields an efficient approach to estimation of the distribution of the $\rho_{DCCA}$ coefficients under the null hypothesis.

\section{Method}
\label{sec:method}
The technical details of our method are catalogued in Algorithm~\ref{alg:test}. Here we provide an informal description and explanation of the major steps involved. 

Our method consists of two steps: firstly we use theoretical results to calculate an approximation to the distribution of the $\rho_{DCCA}$ coefficients under the null hypothesis 
of independence.
Secondly we construct a test-statistic and rejection region which incorporate on the one hand the fact that the Hurst exponents of the two time-series
are known only approximately through estimation and
on the other hand assumptions as to the orders of magnitude over which scaling is expected under the alternative hypothesis of power-law cross-correlation.

\subsection{The asymptotic distribution of the DCCA correlation coefficients}

We showed in the previous section that it is desirable to use the \emph{joint} distribution of the $\rho_{DCCA}$ coefficients in order to judge long-range dependence.
The aim of this section is to describe the theory which allows us to estimate the probability that $\rho_{DCCA}(n_1)>a_1 \& \dots \& \rho_{DCCA}(n_r)>a_r$ for some choice of $a_1, \dots, a_r$ under the null-hypothesis of long-range independence and the steps required in practice for execution of this estimation.

To recap the first step in any DCCA analysis is to form:
\begin{equation}
X_j(t) = \sum_{i=1}^t Y_j(i) \label{eq:subsample}
\end{equation}

Because $X_j(t)$ involves a sum over samples, as $t$ grows, the distribution of $X_j(t)$ may be shown, under certain assumptions, to converge at low frequencies to a specific class of Gaussian process, namely \emph{fractional Brownian motion} (see \cite{taqqu1975weak} and Proposition~\ref{prop:Taqqu}). This may moreover be shown 
to be true for the \emph{joint} distribution of $X_1(t)$ and $X_2(t)$, which may be shown to converge to a bivariate fractional Brownian motion \cite{marinucci2000weak}.  An implication of this convergence which may be shown in simulations and in theory (Proposition~\ref{prop:non_gauss}) is that the distribution of the DFA and DCCA coefficients of the two time-series converge, at large scales, to the distribution of the DFA and DCCA coefficients of a bivariate fractional Brownian motion (resp. bivariate fractional Gaussian noise). Thus if we succeed in calculating the distribution of the DFA and DCCA coefficients of a bivariate fractional Brownian motion, then these will approximate the distribution of the coefficients of a broad
class of long-range dependent time-series.

Thus the next step is to consider the distribution of the DFA and DCCA coefficients specifically for fractional Gaussian noise time-series. The asymptotics of the DFA coefficients in this context have already been studied by \cite{DFA_asymp}. 
In particular, the authors of \cite{DFA_asymp} show that as $[N/n_r] \rightarrow \infty$ then the DFA coefficients of a fractional Brownian motion are normally distributed. 
We extend
this work by showing that this limit generalizes to the DCCA coefficients under the null hypothesis of independence (Proposition~\ref{prop:semi}). Thus as $[N/n_r] \rightarrow \infty$ the DCCA coefficients are normally distributed, assuming independence.
These two asymptotic analyses imply that the $\rho_{DCCA}$ coefficients of a bivariate fractional Gaussian noise are \emph{also} asymptotically normal (Proposition~\ref{prop:CLT}). The covariance matrix of the normal limit may be exactly calculated using the covariance and means
of the DFA and DCCA normal limits and these moments may themselves be approximated in a tractable and accurate manner. The convergence rate is the following:

\begin{equation}
\left(\sqrt{N/n_1})\rho_{DCCA}(n_1),\dots, \sqrt{N/n_r}\rho_{DCCA}(n_r)\right) \rightarrow \mathcal{N}(0, \Theta(H,G)) \label{eq:central_limit}
\end{equation}

Thus we obtain an approximation to the distribution of the $\rho_{DCCA}$ coefficients of a bivariate fractional Gaussian noise; in virtue of the fact that the DCCA coefficients tend in distribution to the DCCA coefficients of fractional
Gaussian noise at large scales, this central limit is thus also an approximation to the distribution of the $\rho_{DCCA}$ coefficients
of $Y_1$ and $Y_2$ across the semi-parametric class for large scales (Proposition~\ref{prop:glue}). The accuracy of this approximation may be shown in simulations to be sufficient to be \emph{useful} in practice (see Sections~\ref{sec:check_fGn},~\ref{sec:check_non_gauss} and~\ref{sec:check_spurious}).
So given that we have access to $H$ and $G$,  we may calculate the probability under the null hypothesis that $\rho_{DCCA}(n_1)>a_1 \& \dots \& \rho_{DCCA}(n_r)>a_r$.

\subsection{Upper bounding the quantiles of the distribution of the DCCA correlation coefficients}

We now deal with the question: what if $H$ and $G$ are known only approximately through estimation? The solution we provide is to calculate a covariance matrix $C$ so that if $0.5\leq H,G<1$ then asymptotically:

\footnotesize
\begin{eqnarray}
&& \text{Pr}(\rho_{DCCA}(n_1)>a_1 \& \dots \& \rho_{DCCA}(n_r)>a_r) \\
&<& \text{Pr}(x_1>a_1 \& \dots \& x_r>a_r | (x_1,\dots,x_r) \sim \mathcal{N}(0,C)) \label{eq:upper_bound}
\label{eq:upper_bound}
\end{eqnarray}
\normalsize

$C$ is calculated by choosing the correlations between dimensions to be identical to the maximum correlations between $\rho_{DCCA}$ coefficients considering $H$ and $G$ in this range and the diagonal entries are chosen so 
that these are equal to the maximum variances for this range. This choice then yields Equation~\eqref{eq:upper_bound} in the asymptotic regime.
Given that more exact information as to the magnitude of the true $H$ and $G$ is available, $C$ may be recalculated to yield a more powerful test.

\subsection{The test-statistic and rejection region}

The test-statistic is defined as follows. If $H$ and $G$ are known exactly then $C$ is given by the limiting covariance of the central limit approximation of Equation~\eqref{eq:central_limit}. Otherwise, $C$ is calculated
according to Equation~\eqref{eq:upper_bound}. Then, for a fixed $\kappa$:

\small
\begin{multline}
\mathcal{T}_\kappa =  \text{max}_{\lambda>0}( \exists n_{r_1} , \dots, n_{r_\kappa}| \\ \rho_{DCCA}(n_{r_1}) > \lambda \sqrt(C_{r_1,r_1}),\dots, \rho_{DCCA}(n_{r_\kappa}) > \lambda \sqrt(C_{r_{\kappa},r_{\kappa}}) \\
\text{ or } \\ -\rho_{DCCA}(n_{r_1}) > \lambda \sqrt(C_{r_1,r_1}),\dots, -\rho_{DCCA}(n_{r_\kappa}) > \lambda \sqrt(C_{r_{\kappa},r_{\kappa}}))
\end{multline}
\normalsize

The parameter $\kappa$ is included because we wish to be able to reject the null-hypothesis of no power-law cross-correlation even if a few of the $\rho_{DCCA}$ coefficients display
no correlation.
Thus the test-statistic incorporates two important aspects: firstly the joint distribution is used ensuring robustness to short range correlations, provided $\kappa$ is close to $r$. Secondly, by choosing $\kappa$ to be close to $r$ but not equal to $r$, 
we include the possibility that, certain scales are either contaminated by confounding noise and thus display \emph{no} correlation or that long-range cross-correlation is only visible over the largest scales; see the discussion and conclusion section
for more on these possibilities (e.g. the delta rhythm in fMRI research and the slow onset of power-law cross-correlation in EEG amplitude time series).

The final step in defining the test involves choosing $\kappa$ and choosing a rejection region for $\mathcal{T}_\kappa$ for a given level $p_{level}$, i.e. we require $\phi$ s.t.:

\begin{equation}
\text{Pr}(|\mathcal{T}_{\kappa}| > \phi\text{ } | \text{ }\mathcal{H}_0) \leq p_{level}
\end{equation}

This may be achieved as follows: if $\kappa = r$ the solution is simple, one uses simply the distribution of $\mathcal{N}(0,C)$. If, however, $\kappa< r$, then we have the following inequality, 
assuming that $n_1,\dots,n_r$ are evenly log-spaced then asymptotically (for large $n_i$ and $N/n_i$):

\small
\begin{multline}
\text{Pr}\left(\rho_{DCCA}(n_{r_1}) > \lambda \sqrt(C_{r_1,r_1}),\dots, \rho_{DCCA}(n_{r_\kappa}) > \lambda \sqrt(C_{r_{\kappa},r_{\kappa}} \right)
\\ \leq \text{Pr}\left(\rho_{DCCA}(n_{1}) > \lambda \sqrt(C_{1,1}),\dots, \rho_{DCCA}(n_{\kappa}) > \lambda \sqrt(C_{{\kappa},{\kappa}} \right)
\end{multline}
\normalsize
Thus: 
\small
\begin{multline*}
\text{Pr}(|\mathcal{T}_{\kappa}| > \phi \text{ } | \text{ } \mathcal{H}_0 ) \\
< 2 \begin{pmatrix}r \\ \kappa \end{pmatrix} \text{ } \text{Pr}\left(\rho_{DCCA}(n_{1}) > \phi \sqrt(C_{1,1}),\dots, \rho_{DCCA}(n_{\kappa}) > \phi \sqrt(C_{{\kappa},{\kappa}}) \right)
\end{multline*}
\normalsize
%
%
%

The test procedure is implemented in software (described in Section~\ref{sec:software} of the Appendix) which 
is available for download\footnote{The software implementing the test in MATLAB is available at \url{http://www.user.tu-berlin.de/blythed/DCCA_matlab}}.

A final question which may arise is the following: why does rejection of the test imply long-range cross-correlation and not simply cross-correlation? The answer is that it is possible to show that the bivariate time-series $(X_1(t),X_2(t))^\top$ has asymptotic properties identical to those of a bivariate fractional Brownian motion sharing the Hurst exponents of each of the components \cite{marinucci2000weak}. This implies that the exponent of cross-correlation (Equation~\eqref{eq:LRCC}) is long-range if one of $H$ and $G$ is greater than $1/2$, since if fractional Gaussian noises are correlated then they are power-law correlated provided $H,G \geq 1/2$ and one of $H$ and $G$ is $>1/2$ \cite{SimulatorVfbm}.
This implies that if the
time-series are short-range cross correlated but not long-range cross-correlated, then the approximating fBm will possess independent components and for large enough scales the hypothesis of no long-range cross correlation will not be rejected.

\section{Simulations}

In this section we present simulations which check the proposed method for correctness, power and efficiency. The first two simulations (Sections~\ref{sec:check_fGn},~\ref{sec:check_non_gauss})
check the accuracy of the approximation provided by the central limit theorem (Equation~\eqref{eq:central_limit}); the third (Section~\ref{sec:check_spurious}) checks the robustness
of the test-statistic to spurious short-range correlations; the fourth (Section~\ref{sec:check_ub}) checks the upper bound of Equation~\eqref{eq:upper_bound}; the fifth checks the power of the test as a function of the correlation 
between the components of a bivariate fractional Gaussian noise; the final Simulation (Section~\ref{sec:check_speed}) checks the computational efficiency of the proposed test.

\subsection{Checking the distribution of the test-statistics for fractional Gaussian noise}
\label{sec:check_fGn}
Since our test is based on approximating the distribution of the test-statistics under the null by means of the asymptotics for fractional Gaussian noise, we first check 
the approximation when the time-series are fractional Gaussian noises.  Thus in the case in which $H$ and $G$ are known exactly and in the case in which $\kappa = r$, the 
central limit approximation should become exact for fractional Gaussian noise in the limit of data points and window sizes (see the Appendix).
In each case we use the fractional Gaussian noise generation method of \cite{SimulatorVfbm} to generate $100,000$ bivariate fractional Gaussian noises with $H=0.7, G=0.8$.
We first check the approximation in the body of the support of the test-statistic for two values of $r$ (number of time-scales on which the $\rho_{DCCA}$ coefficient is calculated). 
The results of these two experiments are displayed in the first row of Figure~\ref{fig:bootstrp_confidence_theory_practice}; comparing the top-left panel with the top-right panel, we see that for larger window sizes, the approximation 
is of higher quality but in both cases an approximation to within 0.005 of the target probability is achieved. The increase in accuracy for higher window sizes relates not to convergence to normality but 
to the approximation of the covariance matrix via tabulation, which is more accurate for larger window sizes.

We then investigate two cases in the tails of the distribution of the test-statistic, which are of more relevance in testing, since a typical test requires control of the Type I error rate to values of at most 0.05. 
Thus on the second row of Figure ~\ref{fig:bootstrp_confidence_theory_practice} we see displayed two cases in which the target quantile is of probability approximately 0.03: the first case involves fewer window sizes than the second case.
Both results of these latter two simulations show that the central limit approximation converges more slowly at the tails of the distribution of the test-statistic than in the
body of the support, although agreement is more than adequate to guarantee usefulness in practice. 
In all cases the results are robust to the exact choice of parameter values, $G$ and $H$, and the quality of the approximations made increase in $[N/n_i]$ (where the $n_i$ are the window sizes). Because we use a detrending operation, the results are robust to polynomial trends 
of degree $d$ in $X_1(t)$ and $X_2(t)$. All simulations cited here set the detrending degree to $d=1$, however
additional simulations (not presented here) confirm that the results generalize to higher order detrending.

\begin{figure}[h]
\begin{center}
$\begin{array}{c c }
\includegraphics[width=40mm]{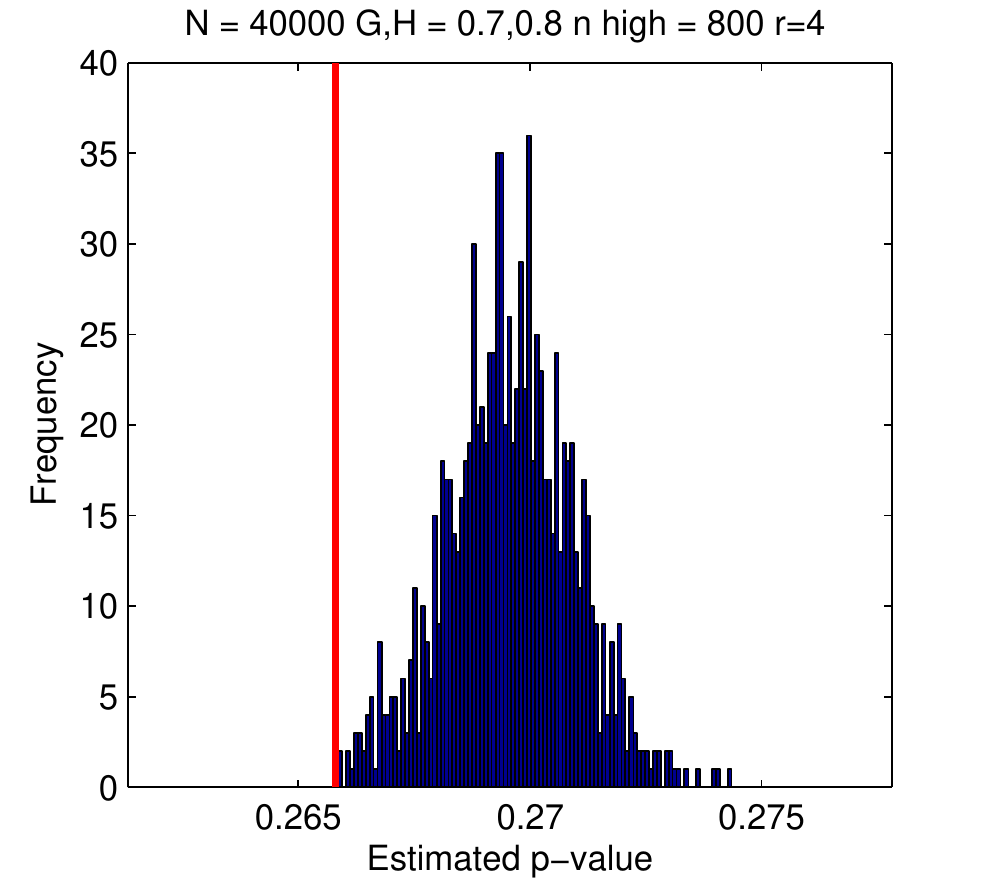} &
\includegraphics[width=40mm]{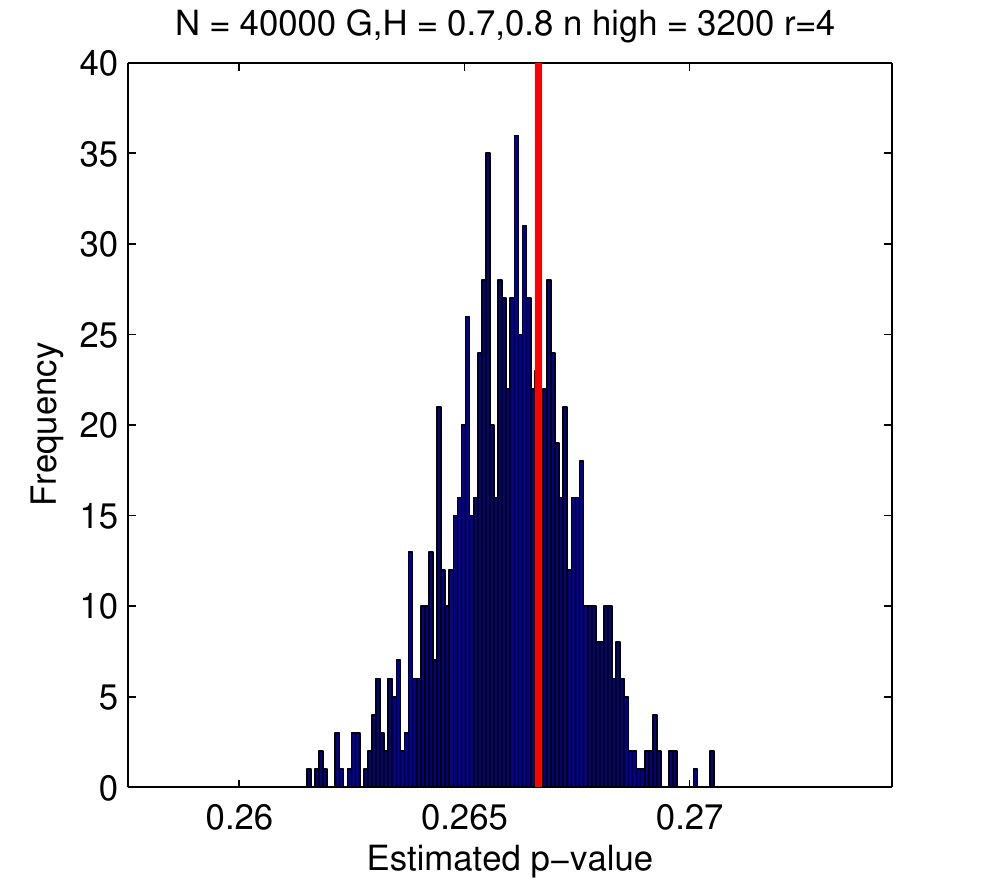}  \\
\includegraphics[width=40mm]{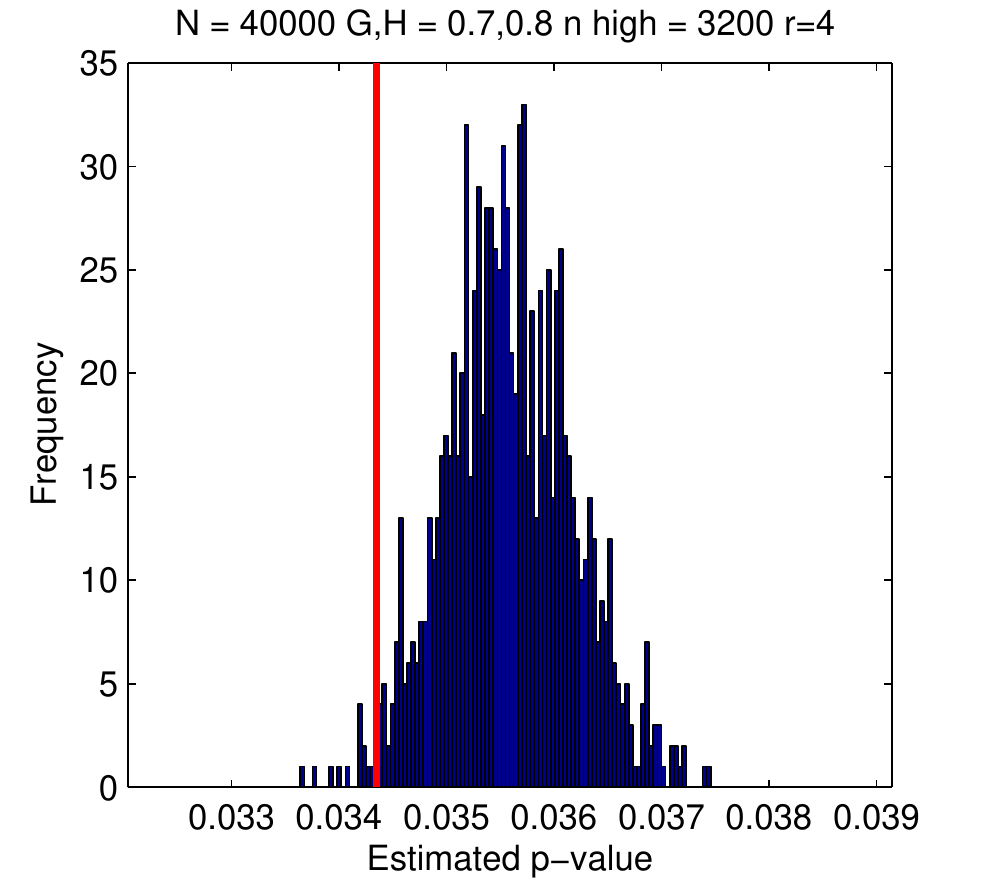} &
 \includegraphics[width=40mm]{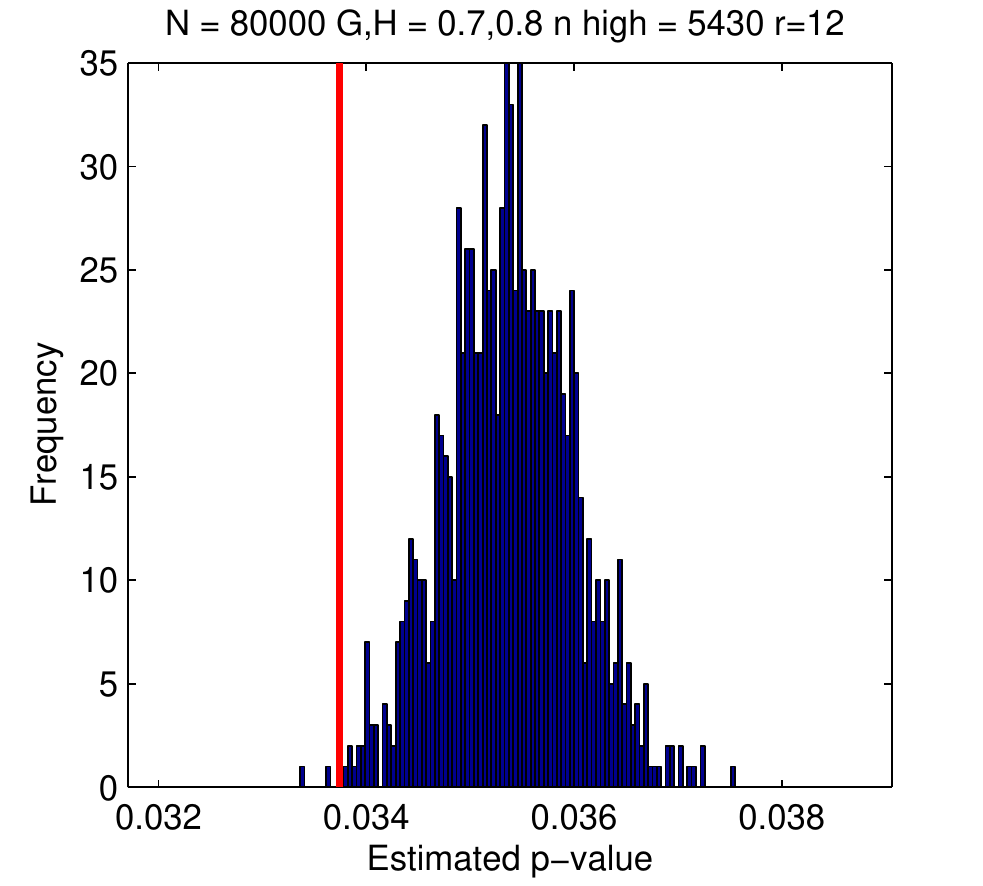}
\end{array}$
\end{center}
\caption{Accuracy of the central limit approximation to the distribution of DCCA correlation coefficients on simulated fractional Gaussian noise time-series; the figure displays the agreement of our theoretically calculated $p$-values with quantiles calculated from simulated fractional Gaussian noise.  Parameter values are given above the figures. In each case, histograms are displayed (given by bootstrapping) to display the variability in the simulated $p$-values, where the red-lines display the theoretically calculated values. We observe a high level of agreement in all cases with marginally slower convergence in the tails of the distribution (bottom middle), where nevertheless, agreement 
lies within $\pm0.002$ when $p\approx0.03$.    \label{fig:bootstrp_confidence_theory_practice}}
\end{figure}

\subsection{Checking the distribution of the test-statistic for non-Gaussian processes}
\label{sec:check_non_gauss}

The second simulation checks the quality of the approximation (Equation~\eqref{eq:central_limit}) for non-Gaussian data and compares the results to the results for Gaussian data. In order to generate an appropriate non-Gaussian time-series we use the result of the paper \cite{blum1994simple};
here it is shown that a non-Gaussian fractional noise may be generated by appropriately filtering a non-Gaussian white noise. Accordingly we generate a super-Gaussian white noise ($x_{nonGaussian} = \text{sign}(x_{Gaussian})\times|x_{Gaussian}|^\phi$)
for $\phi=3$ and filter according to the procedure proposed by \cite{blum1994simple}. The resulting time-series is thus a linear time-series sharing the second-order statistics of fractional-Gaussian noise but differing in its higher order-statistics (see the right
panel of Figure for an illustration of the data).

The 0.05 level calculated on 10,000 samples of this non-Gaussian fractional noise with 
$N=5000,10000,20000,40000$ and $G=0.7,H=0.8$ is compared with the 0.05 level calculated on simulated fractional Gaussian noises and the level generated by the proposed method. 
The results displayed in Figure~\ref{fig:null_test_non_gaussian} show that the accuracy of the theoretical approximation is similar for non-Gaussian fractional noise and Gaussian fractional noise. This may be related
to the fact that the second-order statistics of both processes are similar, and, under the null-hypothesis the covariance matrix which we calculate for the $\rho_{DCCA}$ coefficients
depends only on the second-order statistics of the two processes considered; see Equation~\eqref{eq:second_order_stats} of the Appendix.

\begin{figure}
\begin{center}
$\begin{array}{c c}
\includegraphics[width=40mm,clip=true,trim= 5mm 0mm 2mm 0mm]{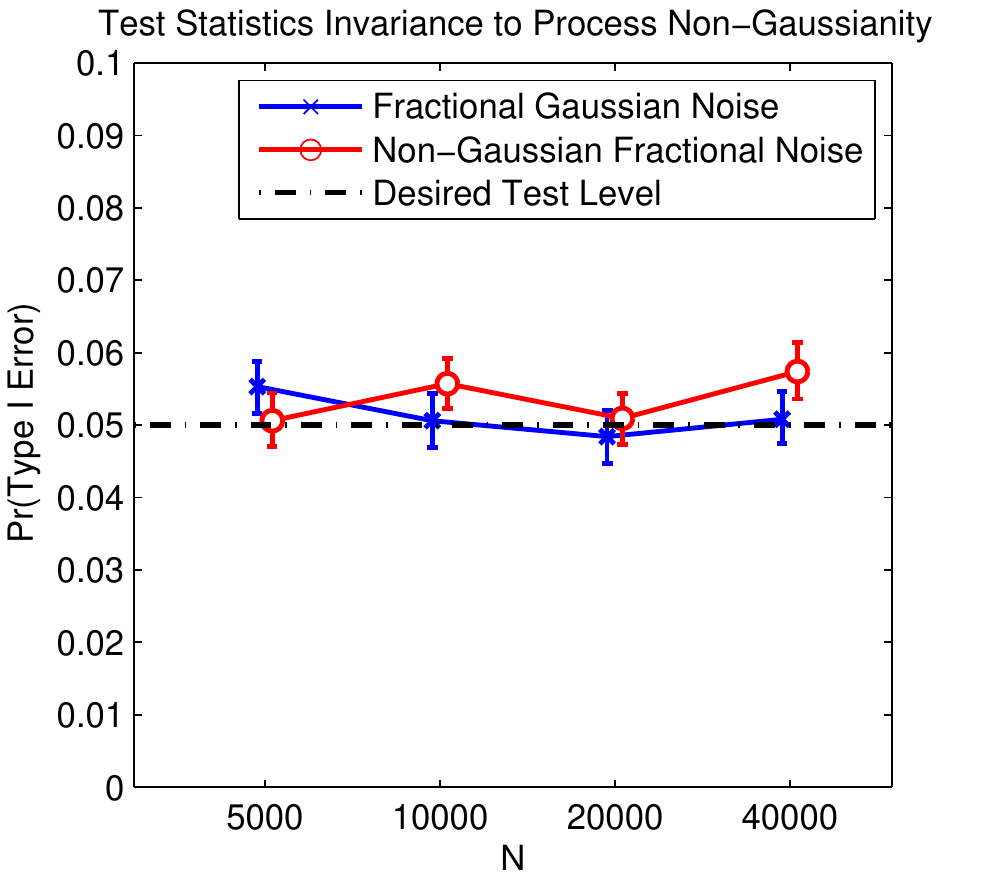} & \includegraphics[width=40mm,clip=true,trim= 5mm 0mm 2mm 0mm]{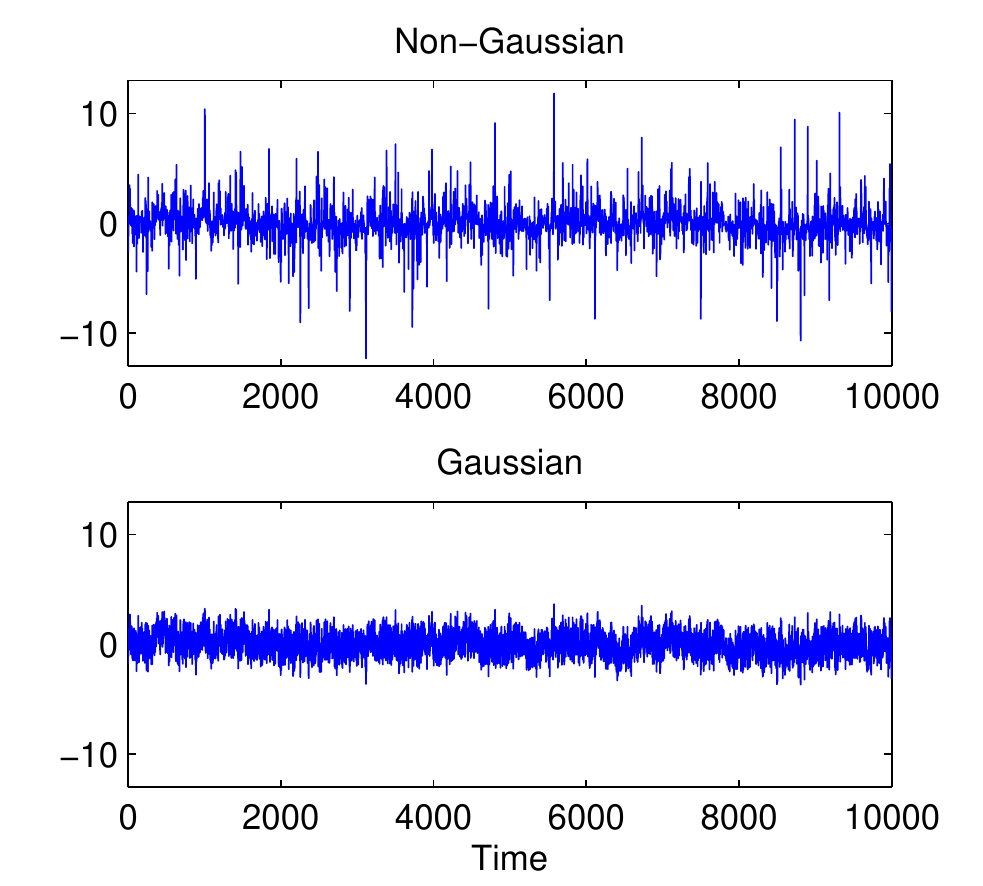}
\end{array}$
\end{center}
\caption{Invariance of the distribution of the test statistic to non-Gaussianity. In the right hand panel the distributions of the two respective time-series tested are displayed. In the 
lower right panel an empirical time-series of the Gaussian model is displayed whilst in the lower panel an empirical time-series of the non-Gaussian model is displayed. The non-Gaussian density thus has fatter tails than the Gaussian.
On the other hand the left hand panel compares the empirical distribution of the $\rho_{DCCA}$ coefficients on these models in comparison to the theoretical prediction (calculated on 10,000 draws from each model), for time-series with  $N = 5000,10000,20000$ and $40000$, $H,G=0.7,0.8$. The results show that the theoretical prediction yields comparable results for both Gaussian and non-Gaussian models.
\label{fig:null_test_non_gaussian}}
\end{figure}

\subsection{Checking the distribution of the test-statistic for Gaussian time-series with a correlated short-range component}
\label{sec:check_spurious}

The third simulation checks the robustness of control over the type I error rate for the proposed method in the presence of a spurious short-range correlated component and compares against a multiple testing-correction. Although
choosing large enough scales reduces the effect of these short-range correlations on the distribution 
of the $\rho_{DCCA}$ coefficients, exactly at what magnitude the smallest scale should be chosen in practice is unclear. In this simulation we investigate a case in 
which the short-range correlations in the time-series exert influence over the distribution of the $\rho_{DCCA}$ coefficients at small scales and evaluate the 
on the distribution of our test-statistic; thus although for increasingly larger windows sizes Propositions~\ref{prop:non_gauss} and~\ref{prop:glue} guarantee
that the central limit theorem provides a more accurate approximation, in the finite window size regime the approximation is coarser. We compare the probability of Type I error
for the test proposed above using the exact asymptotics of $\mathcal{T}_r$ and the multiple testing procedure discussed in Section~\ref{sec:DCCA_testing} with the test-level set for both procedures to 0.05.
(Thus the multiple testing procedure uses the \emph{univariate} asymptotics of the DCCA coefficients and reports LRCC if at least one of the coefficients delivers a significant readout after
a Bonferroni multiple comparisons correction.)
In particular, we 
simulate 10,000 bivariate time-series which are generated as a linear superposition of a bivariate fractional Gaussian noise with $H=G=0.9$ and a correlated
Gaussian White Noise, high pass-filtered to the frequencies above 0.45 times its sampling frequency. This simulation is repeated for 4 log-spaced choices of $N = 5000,10000,20000,40000$
and the probability of rejection at the 0.05 null hypothesis level, calculated using the proposed method compared with the probability of  rejection at the 0.05 null hypothesis using
the multiple testing correction method.

The results are displayed in Figure~\ref{fig:slowly_varying} and show that while the proposed method slightly underestimates the Type I error rate at the 0.05 level ($\approx 0.06$), the multiple testing correction 
grossly underestimates the Type I error rate. Moreover, while the underestimation becomes all the more drastic with increasing $N$ in the case of the multiple testing correction, the
proposed method continues to estimate stably. The reason for this difference is that for the proposed method, the $\rho_{DCCA}$ coefficients at the lowest scales can never be solely
responsible for a rejection of the null-hypothesis: \emph{all} (or most if $\kappa \neq r$) coefficients must display information regarding correlation. On the other hand, since the short-range correlation measured 
in the $\rho_{DCCA}$ coefficients at small scales is measured as significant in these small-scales more often when $N$ is larger (larger sample size), the underestimation of the type I error 
rate grows in $N$ for the multiple testing procedure.

\begin{figure}
\begin{center}
$\begin{array}{c c}
\includegraphics[width=40mm,clip=true,trim= 0mm 0mm 5mm 0mm]{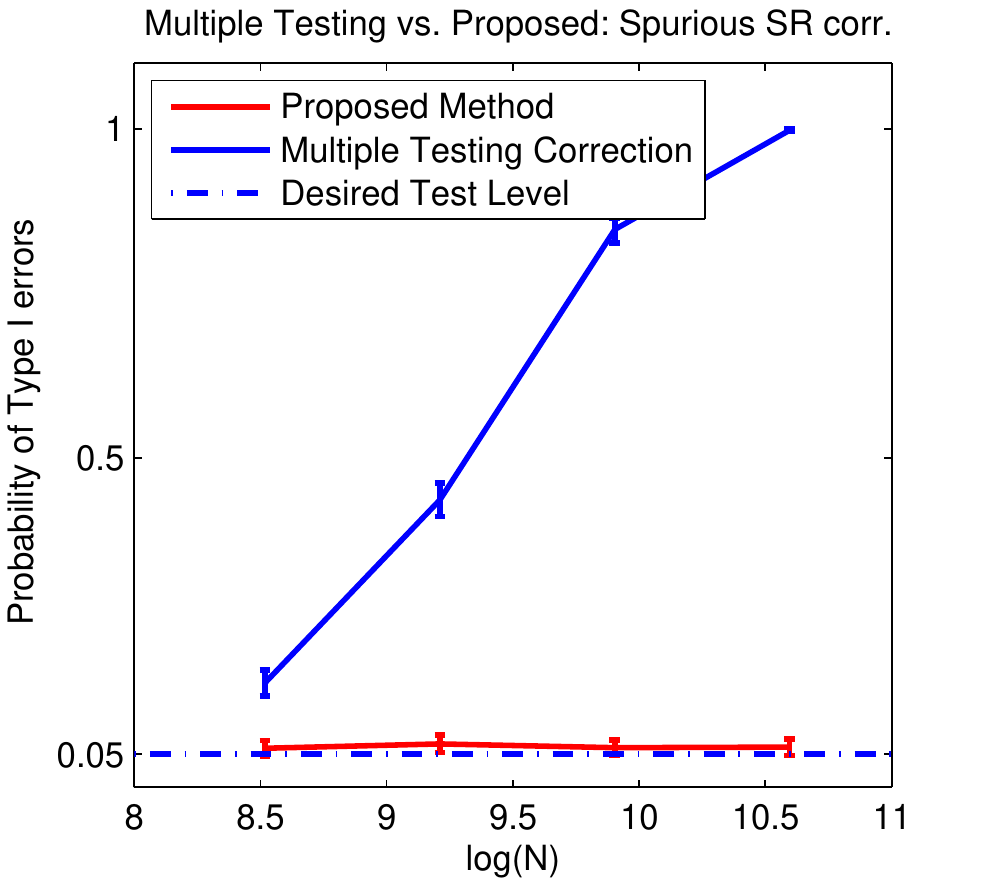} & \includegraphics[width=40mm,clip=true,trim= 0mm 0mm 5mm 0mm]{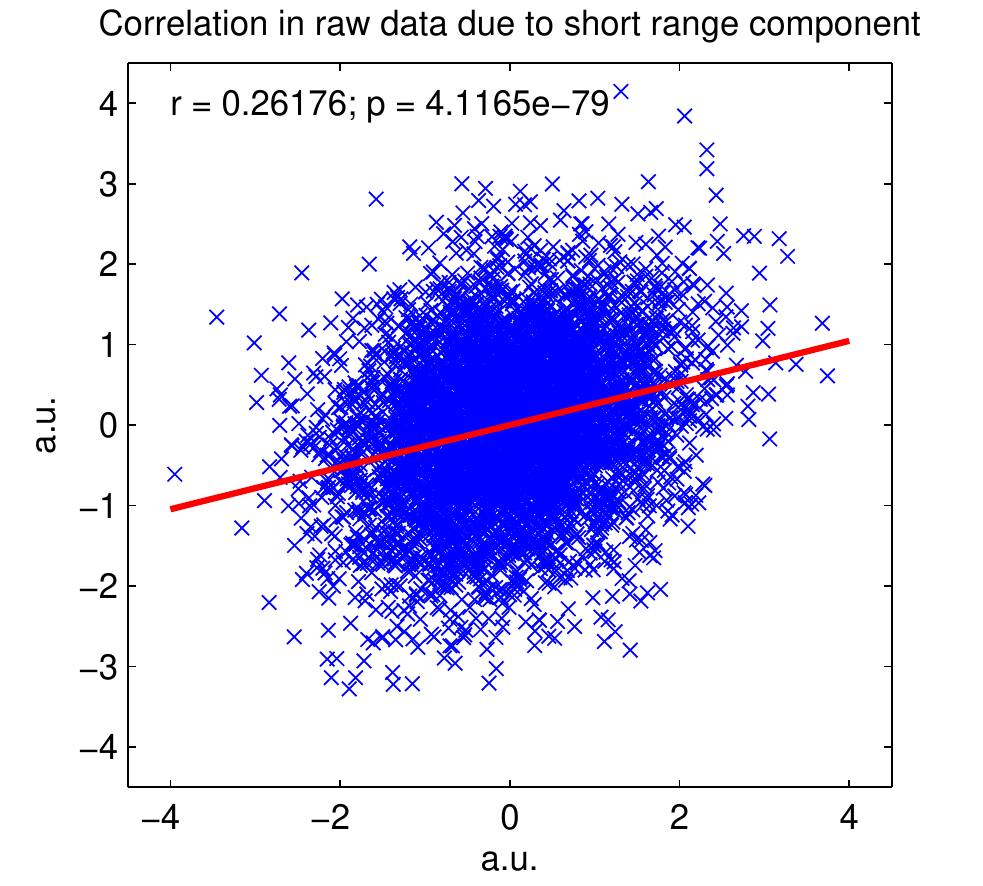}
\end{array}$
\end{center}
\caption{Robustness of the proposed method to short-range correlations; the left-hand panel displays the the probability of rejection of the null-hypothesis using 
a multiple testing correction (blue) of the quantiles of the distribution of the $\rho_{DCCA}$ coefficient on fractional Gaussian noise in comparison to the probability for rejection using the proposed test-statistic.
The proposed test-statistic thus displays a degree of robustness to this rejection not shared by the multiple testing correction. \label{fig:slowly_varying}}
\end{figure}

\subsection{Checking the upper bound for the cases in which the Hurst exponents are estimated}
\label{sec:check_ub}
The fourth simulation checks that the bound we present in Equation~\eqref{eq:upper_bound} may be used to correctly upper bound the probability of a Type I error when $H$ and $G$
are unknown or known only approximately through estimation. To this end we calculate the rejection boundary according to the proposed method when $H$ and $G$ are
known, for $H,G$ in the range $[1/2,1)$, using the central limit theorem and calculation of the limiting covariance matrix and the rejection boundary using the worst case covariance matrix.
The rejection boundaries when $H$ and $G$ are known are plotted in Figure~\ref{fig:check_upper_bound} in blue and the rejection boundary using the worst case upper bound
is plotted in red. The results confirm that the worst case boundary correctly lies above the boundaries when $H$ and $G$ are known, for all values of $H$ and $G$.

\begin{figure}
\begin{center}
\includegraphics[width=60mm]{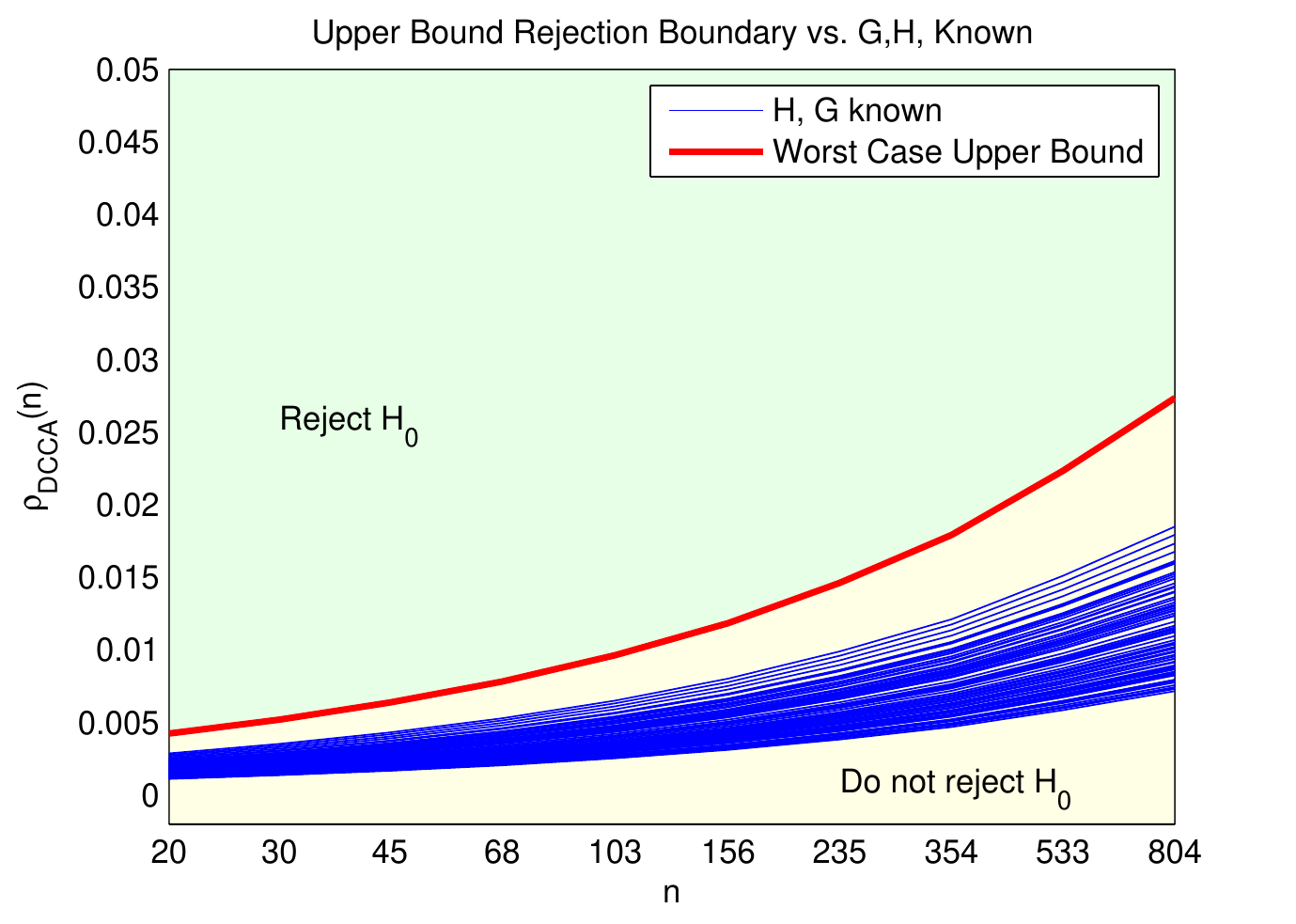}
\end{center}
\caption{Correctness of the $p$-values when $H$ and $G$ are unknown; the figure displays the upper bound rejection line (red) of Equation~\eqref{eq:upper_bound} in comparison to the rejection lines generated by the central limit theorem when $H$
and $G$ are known (blue). Thus when \emph{all} $\rho_{DCCA}$ coefficients lie above the rejection lines, the null is rejected. The fact that the red rejection line lies above every blue
rejection line confirms that the upper bound of Equation~\eqref{eq:upper_bound} allows for strict control over the type I error rate despite access to $H$ and $G$
only through estimation.
\label{fig:check_upper_bound}}
\end{figure}

\subsection{Checking the test is useful by checking test power}

The penultimate simulation tests the power of the test we present when the type I error rate is controlled at the 0.05 level. Thus we check that the test may be effectively used to reject the null in cases of dependence by generating 330 time-series for varying levels of long-range cross-correlation $\rho \in [0.005,0.2]$, $N=40000$, $H,G = 0.7,0.8$. The results are displayed in Figure~\ref{fig:power_function} and show 
that the null-hypothesis is rejected in more than 50 percent of cases when $\rho > 0.07$. Although this shows that the proposed method results in less frequent rejection of the null than in a standard correlation 
analysis (not taking long-range dependence into account), the power is sufficient for detection of weak long-range correlation when $N \geq 40000$ and possesses the advantages of robustness
to polynomial trends and to the presence of confounding short-range correlations.

\begin{figure}[h]
\begin{center}
\includegraphics[width=80mm]{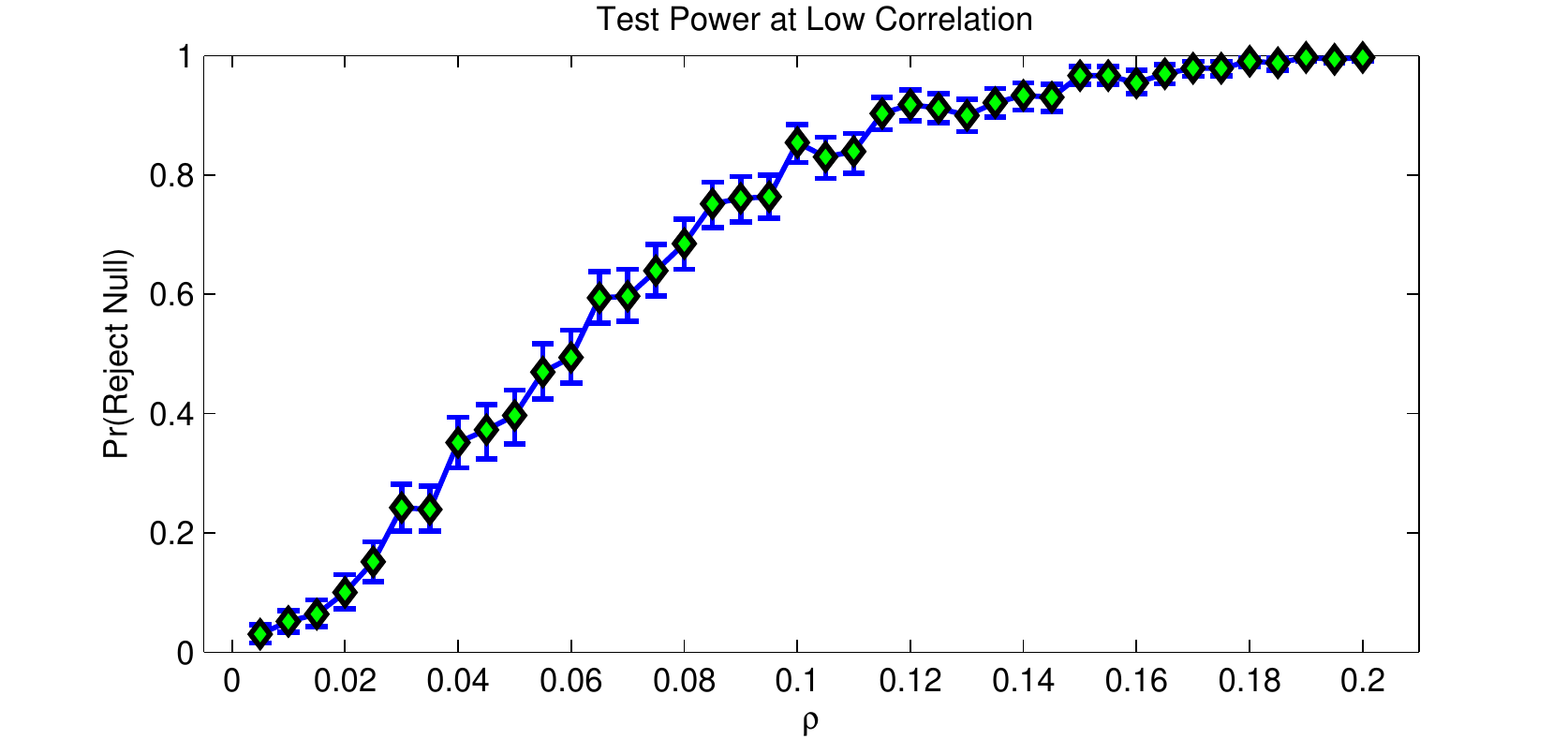}
\end{center}
\caption{Test power of the proposed method; the figure displays the test power plotted as a function of the correlation parameter $\rho$ of the bivariate fractional Gaussian noises which are generated.
The errorbars are calculated using bootstrapped means after simulation of 330 time-series for each value of $\rho$. The results show that for values of $\rho \geq 0.07$,
the null-hypothesis is rejected at the 0.05 level in over 50\% of cases.
\label{fig:power_function}}
\end{figure}  

\subsection{Checking the test is useful in terms of computational efficiency}
\label{sec:check_speed}
Finally we investigate the speed of the proposed method in comparison to computation of the test-distribution via simulated data as per \cite{podobnik2011statistical}. Here we estimate in simulation the amount
of time necessary for stable estimation of the 0.05 quantiles of the test distributions, given that $H$ and $G$ are known using the proposed method and simulation using a fractional 
Gaussian noise generator. Podobnik et al.~require 10,000 samples for this calculation in simulation. Thus in Figure~\ref{fig:speed} we display the amount of
computation time required for the proposed method and the simulation method. The results show that the proposed method is over 5 orders of magnitude faster than the method of
simulation. Although the proposed method involves sampling Gaussian data in order to evaluate the quantiles of the central limit, the simulation shows that drawing these
samples requires considerably less computation time than drawing samples explicitly from a fractional Gaussian noise simulator.

\begin{figure}[h]
\begin{center}
\includegraphics[width=60mm]{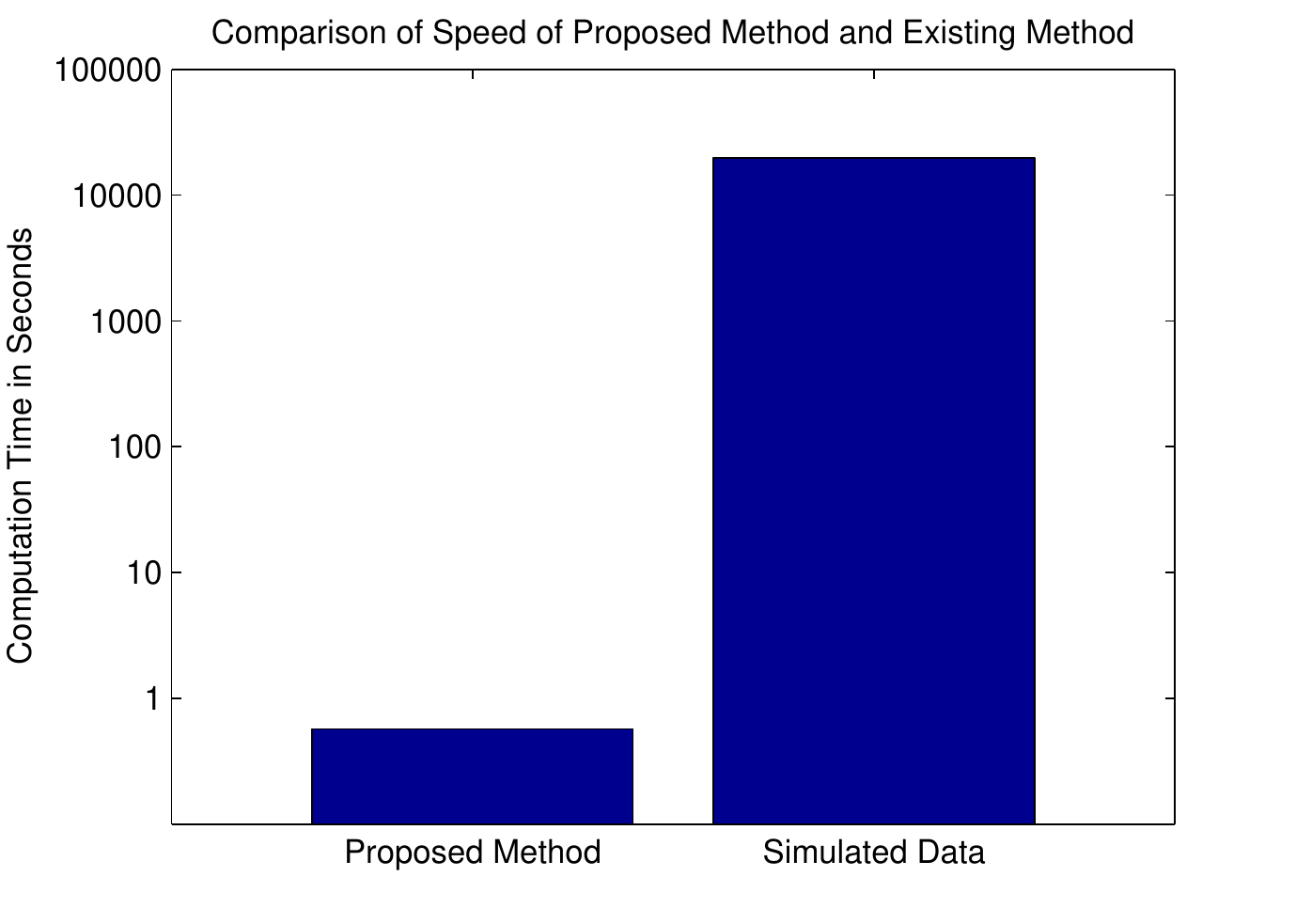}
\end{center}
\caption{Speed-up the proposed method over simulation of surrogate data; the respective values display
the estimated time required for estimation of the distribution of the $\rho_{DCCA}$ coefficients for a time-series of length $N = 80,000$. \label{fig:speed}}
\end{figure}  

\section{Conclusion}

This paper has detailed a test for power-law cross-correlated behaviour which incorporates information across time scales, may be executed efficiently without extensive simulation and 
has been shown to be robust to a range of distributional assumptions. Thus for applications
where interactions are weak we may rigorously check the significance of the observed $\rho_{DCCA}$ coefficients across time-scales. Moreover, because the derivations are based on the $\rho_{DCCA}$ coefficients, the method inherits the advantages possessed by DCCA, i.e. robustness to polynomial trends of a predetermined degree.

Further work will involve broader application of the test in physics, geophysics, biomedical fields and neuroscience, extension of the theory to cover a broader semi-parametric class, for example incorporating the assumption of non-linearity, investigation of alternative modes of detrending and
the development of efficient numerical methods for calculation of the limiting covariance matrix.

\section{Acknowledgements}

The author would like to thank Daniel Bartz for his comments and suggestions on the topic of this paper and Prof. Klaus-Robert M\"uller for valuable advice and support. Duncan Blythe was supported by a grant from the DFG research training group GRK 1589/1 "Sensory Computation in Neural Systems".

\pagestyle{empty}
\bibliographystyle{ieeetr}
{\small
\bibliography{../../Bibs/DFA}}

 \section*{Appendices}
\appendix

\section{Model assumptions and multivariate fractional Brownian motion}

Our model time-series obey the following: we assume we are given two stationary linear processes $Y_1(t) = Y_1^a(t) + Y_1^b(t)$ and $Y_1(t) = Y_1^a(t) + Y_1^b(t)$ where $Y_1^a(t)$ and $Y_2^a(t)$ each have auto covariance functions of the form: $L_1(t)t^{2H-2}$ and $L_2(t)t^{2G-2}$ where the $L_i$ are slowly varying functions at infinity\footnote{ $L_i(t)/L_i(ct) \rightarrow 1$ as $t \rightarrow \infty$ for any $c>0$.}  where the Hurst exponents $H$ and $G$ lie in $[1/2,1)$ \cite{karamata1930mode,lamperti1962semi} and the time-series $Y_i(b)(t)$ 
are deterministic trends of a fixed polynomial order $d$. This formulation allows for non-Gaussianity, for confounding non-stationarities and for the presence of a high frequency component of the power-spectrum differing
at those high frequencies from a power law\footnote{Notice that this semi-parametric class includes the the class analysed by Bardet el al. \cite{DFA_asymp}, viz. stationary Gaussian time-series with an auto covariance function proportional to $k^{2H-2}(1+\mathcal{O}(1/k^\beta))$ since $1+\mathcal{O}(1/k^\beta)$
is slowly varying at infinity.}; this formulation formalizes the formulation of Equations~\eqref{eq:LRTC} and~\eqref{eq:LRCC} of the main body of the paper. We believe that the class of functions we consider is, moreover, more restrictive than necessary and conjecture that the requirement of linearity may be replaced by a constrained non-linearity assumption.

Given this formulation, in the following, we lever existing theoretical work and use the covariance function of fractional Brownian motion in order to construct a test valid for this semi-parametric class (see Section~\ref{sec:theory}).
The covariance function of multivariate fractional Brownian motion is given as follows; if $X_1(t)$ and $X_2(t)$ are two components of the multivariate fractional Brownian motion, then:

\begin{enumerate}
\item If $H + G \neq 1$
\begin{eqnarray*}
\mathbb{E}(X_1(s)X_2(t)) &=& \frac{\sigma_1\sigma_2}{2} \bigg((\rho + \eta \text{sign}(s))|s|^{H+G} + (\rho - \eta \text{sign}(t) )|t|^{H + G} -  (\rho - \eta \text{sign}(t-s)) |t-s|^{H + G}\bigg) \label{eq:cov}
\end{eqnarray*}
\item  If $H + G = 1$ 
\begin{eqnarray*}
\mathbb{E}(X_1(s)X_2(t)) &=& \frac{\sigma_1\sigma_2}{2} \bigg(\rho(|s| + |t| - |t-s|) \\
+ \eta(t\text{log}(t) + s\text{log}(s) - (t-s)\text{log}(t-s))\bigg)
\end{eqnarray*}
\end{enumerate}

Here $H$ and $G$ are the Hurst exponents of $X_1$ and $X_2$ respectively.
Under certain regularity conditions Lavancier et al.~\cite{LavCov} are able to show that this covariance function moreover applies to any $L^2$ self-similar multivariate process with stationary increments.

\section{Method details}

Algorithm~\ref{alg:test} formalizes the proposed test procedure; subroutines are stated explicitly in Algorithms~\ref{alg:wccov} and~\ref{alg:critical_region} respectively.

\begin{algorithm}[h]
\caption{Worst case covariance}
\label{alg:wccov} 
\begin{algorithmic}[1]
	\Function{wcCov}{$n_1, \dots, n_r < N$, $d = \text{trend degree},H_{low},H_{high},G_{low},G_{high}$}
		\For{$i = \{ 1,\dots, r\}$}
			\State $C_{i,i} = \text{max}_{H \in [H_{low},H_{high}], G \in [G_{low},G_{high}]} \text{var}(\rho_{DCCA}(n_i,X_G,X_H))$
		\EndFor 
		\For{$i \neq j$}
			\State $r_{i,j} = \text{max}_{H \in [H_{low},H_{high}], G \in [G_{low},G_{high}]} \times \newline\hspace{1in} \text{corr}(\rho_{DCCA}(n_i,X_G,X_H),\rho_{DCCA}(n_j,X_G,X_H))$
			\State $C_{i,j} = r_{i,j} \sqrt{C_{i,i}} \sqrt{C_{j,j}} $
		\EndFor 
		\State \Return $C$  
	\EndFunction
\end{algorithmic}
\end{algorithm}

\begin{algorithm}[h]
\caption{Calculate the critical region}
\label{alg:critical_region} 
\begin{algorithmic}[1]
	\Function{critRegion}{$C$,$p$, $\kappa$}
		\For{$i = \{ 1,\dots, r\}$}
			\State $a_i(\theta) = \theta  \sqrt{C_{i,i}}$
		\EndFor 
		\State $A(\theta)$ = $\{ (x_1,\dots,x_r) | x_{r_1} >a_{r_1}(\theta),\dots,x_{r_{\kappa}} > a_{r_\kappa}(\theta) \text{ or } -x_{r_1} >a_{r_1}(\theta),\dots,-x_{r_{\kappa}} > a_{r_\kappa}(\theta) (\text{some } r_1, \dots) \}$.
		\State $A = A(\theta^*)$ s.t. \newline $\theta^* = \text{min}_\theta\left(Pr_{\mathcal{N}(0,C)}((x_1,...,x_r) \in A(\theta)) < p\right)$.
		\State \Return $A$  
	\EndFunction
\end{algorithmic}
\end{algorithm}

\begin{algorithm}[h]
\caption{Test for power-law interdependence}
\label{alg:test} 
\begin{algorithmic}[1]
	\Function{statDCCA}{$Y_1,Y_2, n_1, \dots, n_r$, $d$, $H_{high}$, $H_{low}$, $G_{high}$, $G_{low}$, $p$, $\kappa$}
		\State Compute $X_1$ and $X_2$ by $X_1(t) =  \sum_{j=1}^{t} Y_1(j)$ and  $X_2(t) =  \sum_{j=1}^{t} Y_2(j)$
		\For{$i = \{ 1,\dots, r\}$}
			\State Split $X_1$ and $X_2$ into windows of length $n_i$: $X_{j,n_i}^{1},\dots X_{j,n_i}^{[N/n_i]}$ and 
			\State Calculate the degree $d$ polynomial trend in each window to yield $\widehat{X}_{j,n_i}^{1},\dots \widehat{X}_{j,n_i}^{[N/n_i]}$.
			\State Calculate $F_{k,X_1,X_2}^2(n_i) = \widehat{\mathbb{E}}\left((X_{1,n_i}^{k}-\widehat{X}_{1,n_i}^{k})(X_{2,n_i}^{k}-\widehat{X}_{2,n_i}^{k})\right)$
			\State Calculate $F^2_{X_1,X_2}(n_i) = \widehat{\mathbb{E}}(F_{k,X_1,X_2}^2(n_i))$.
		    \EndFor
		\State $C$ = {\scriptsize{WC}\normalsize C\scriptsize OV\normalsize}$\left(n_1, \dots, n_r < N, d,H_{low},H_{high},G_{low},G_{high}\right)$
		\State $A$ = {\scriptsize{CRIT}\normalsize R\scriptsize EGION\normalsize}($C$, $p$, $\kappa$)
		\If{ $(F^2(n_1),\dots,F^2(n_r)) \in A$}{ \Return \emph{reject} $\mathcal{H}_0$}
 		\Else{ \Return \emph{do not reject}} $\mathcal{H}_0$  \EndIf
	\EndFunction
\end{algorithmic}
\end{algorithm}
\subsection{Technical Issues}
\subsubsection{Calculation of probabilities}

Algorithm~\ref{alg:critical_region} involves the calculation of a corresponding
Gaussian integral. Numerical calculation of this integral is inefficient and inaccurate when $r$ is large. In such cases, it is more straightforward and more efficient
to estimate the relevant probabilities using a Gaussian random number generator\footnote{We use {\tt mvnrnd} in MATLAB.}. Notice that this is a considerably more compact computation than the computation undertaken
by Podobnik et al. \cite{podobnik2011statistical}. When $r=25$, $10^6$ samples suffice to obtain a stable estimate for a case in which the corresponding probability is $\approx 0.0510$ to within a tolerance of $\pm 0.001$ which takes approximately $2.4$ seconds in MATLAB. The efficiency in estimation is due to the high correlations between the DCCA coefficients across scales, which reduces the effective dimensionality of the support of the distribution.

\subsubsection{Calculating asymptotic covariance}

We calculate the variance of the $F_{X_1,X_2}^2(n)$ by evaluating, for $H,G = 0.5,0.52,\dots,0.96,0.98$ an exact formula (Equation~\ref{eq:covariance_fun}) for large but computationally feasible $N$ and $n$. 
We then calculate $[N/n]$ times this variance which should correspond up to finite sample error to the asymptotic limit posited by the theory (below). Then for a new $N'$ and $n'$ the required variance
is approximately $1/[N'/n']$ times this limit. 

Moreover, in order to calculate the cross terms we tabulate the correlation between $F_{X_1,X_2}^2(n), F_{X_1,X_2}^2(n')$ for $H,G = 0.5,0.52,\dots,0.96,0.98$ and for feasible $N$ and $n$ and $n'=0.01\times n,  \dots, 0.98 \times n, 0.99 \times n$.
For $n' < 0.01 \times n$ the correlation is approximately zero but is upper bounded by the correlation when $n'=0.01$, to be conservative. The covariance term of $F^2(n)$ is then calculated from the variance and correlation terms.

Finally we require the mean of the $F^2_{X_1,X_1}(n)$ terms. For polynomial trending of degree one this mean is given analytically (see ref. \cite{DFA_asymp}, Property 3.1). 
Otherwise the mean may be calculated in a similar manner to the calculation of covariance.
Given these means and covariance matrix, the theory in the appendix yields the full covariance matrix of the DCCA correlation coefficients. 

%
\section{Derivation of covariance expressions}
\label{sec:derivation}
\subsection{Form of the correlation function between time varying DCCA (DFA) coefficients}
Let $P$ be the orthonormal projection onto the subspace of $\mathbb{R}^n$ spanned by $(1,1, \dots,1), \dots, (1,2^i \dots,n^i), \dots ,(1,2^d,\dots, n^d)$ and define
the subspace $E_j^d$ as the span of $(1,1, \dots,1), \dots, (((j-1)n+1)^i,((j-1)n+2)^i \dots,((j-1)n+n)^i), \dots , (((j-1)n+1)^d,((j-1)n+2)^d \dots,((j-1)n+n)^d)$. Then $PX$ is the least 
squares estimate of the polynomial trend of degree $d$ on $X$. Then define $Q:=I-P$. Then, following~\cite{DFA_asymp} we have:

\begin{eqnarray*}
&& \text{cov}(F^2_{1,X,X}(n),F^2_{j,X,X}(n)) \\
 &=& \frac{1}{n^2}\text{cov}\left( (X^{(1)} -  \mathbb{P}_d(X^{(1)}))^\top (X^{(1)} -  \mathbb{P}_d(X^{(1)})), (X^{(j)} -  \mathbb{P}_d(X^{(j)}))^\top (X^{(j)} -  \mathbb{P}_d(X^{(j)})) \right)\\
&=& \frac{1}{n^2}\text{cov}\left( (X^{(1)} -  PX^{(1)})^\top (X^{(1)} -  PX^{(1)}), (X^{(j)} - P X^{(j)})^\top (X^{(j)} - PX^{(j)}) \right) \\
&=&  \frac{1}{n^2}\text{cov}\left( (QX^{(1)})^\top (QX^{(1)}), (QX^{(j)})^\top (QX^{(j)}) \right) \\
&=&  \frac{1}{n^2}\mathbb{E}\left( (QX^{(1)})^\top (QX^{(1)}) (QX^{(j)})^\top (QX^{(j)}) \right) -  \frac{1}{n^2}\mathbb{E}\left( (QX^{(1)})^\top (QX^{(1)})\right)\mathbb{E}\left( (QX^{(j)})^\top (QX^{(j)})\right)\\
&=&  \frac{1}{n^2}\mathbb{E}\left(\sum_i (QX^{(1)})_i^2 \sum_k (QX^{(j)})_k^2 \right) -  \frac{1}{n^2}\mathbb{E}\left(\sum_i (QX^{(1)})_i^2\right) \mathbb{E}\left(\sum_k (QX^{(j)})_k^2\right) \\
&=& \frac{1}{n^2}\sum_{i,k} \mathbb{E}\left((QX^{(1)})_i^2 (QX^{(j)})_k^2\right) -  \frac{1}{n^2} \sum_i \mathbb{E}\left((QX^{(1)})_i^2\right) \sum_k \mathbb{E}\left( (QX^{(j)})_k^2\right) \\
&=& \frac{2}{n^2}\sum_{i,k} \left(\mathbb{E}\left((QX^{(1)})_i (QX^{(j)})_k\right)\right)^2 \\
&=& \frac{2}{n^2} \text{trace}\left(Q\Sigma^{1,j}Q \times (Q\Sigma^{1,j}Q)^\top \right) \\
&= & \frac{2}{n^2} \text{trace}\left(Q\Sigma^{1,j}Q(\Sigma^{1,j})^\top Q \right)  \\
&= & \frac{2}{n^2} \text{trace}\left((I-P)\Sigma^{1,j}(I-P)(\Sigma^{1,j})^\top \right) \\
\end{eqnarray*}

The transition from the sixth to the seventh line is justified by Isserlis's theorem.  Line seven is a 
Froebenius norm, thus justifying the transition to line 8. The transition to the penultimate line uses the fact that $Q$ is 
symmetric. The transition to the last line uses $\text{trace}(AB) = \text{trace}(BA)$ and $Q^2=Q$.

Similarly, under the null hypothesis of zero correlation we obtain the following form for the cross-covariance function:

\begin{equation}
\text{cov}(F^2_{j,X_1,X_2}(n),F^2_{j',X_1,X_2}(m)) = \frac{1}{n^2} \text{trace}\bigg((I_m-P_m)\Sigma_H^{j,j'} \\  \times (I_n-P_n)(\Sigma_G^{j,j'})^\top \bigg) \label{eq:covariance_fun}
\end{equation}

$\Sigma_H^{j,j'}$ denotes the covariance matrix between the $j^\text{th}$ window of size $n$ and the $j^\text{th}$ window of size $m$ of a fractional Brownian motion with Hurst parameter $H$.
\newpage
\subsection{Form of the cross-covariance between the DFA and DCCA coefficients and invariance of the covariance under the null to non-Gaussianity}
\label{sec:delta}
This is necessary to use the delta method in order to calculate the covariance function of $\rho_{DCCA}(n)$.

\begin{eqnarray*}
\text{cov}(F_{1,X_1,X_2}(n), F_{j,X_1,X_1}(n)) &=& \frac{1}{nm}\text{cov}((PX_1^{(1)})^\top PX_2^{(1)},(PX_1^{(j)})^\top PX_1^{(j)}) \\
&=& \mathbb{E}((PX_1^{(1)})^\top PX_2^{(1)}(PX_1^{(j)})^\top PX_1^{(j)}) - \mathbb{E}((PX_1^{(1)})^\top PX_2^{(1)})\mathbb{E}((PX_1^{(j)})^\top PX_1^{(j)}) \\
&=& \sum_{i,k} \mathbb{E}(PX_{1,i}^{(1)} PX_{2,i}^{(1)}PX_{1,k}^{j} PX_{1,k}^{(j)}) -  \mathbb{E}(PX_{1,i}^{(1)} PX_{2,i}^{(1)})\mathbb{E}(PX_{1,k}^{(j)} PX_{1,k}^{(j)}) \\
&=&  2 \sum_{i,k} \mathbb{E}(PX_{1,i}^{(1)} PX_{1,k}^{(j)}) \mathbb{E}(PX_{2,i}^{(1)} PX_{1,k}^{(j)}) \\
&=& 2 \times \text{trace}(P\Sigma_{X_1}^{1,j}P(\Sigma_{X_1,X_2}^{1,j})^\top)
\end{eqnarray*}

Thus in the null hypothesis case, the DCCA and DFA coefficients are uncorrelated.

\begin{eqnarray*}
\text{cov}(F^2_{1,X_1,X_1}(n), F^{2}_{j,X_2,X_2}(n)) &=& \frac{1}{nm}\text{cov}((PX_2^{(1)})^\top PX_2^{(1)},(PX_1^{(j)})^\top PX_1^{(j)}) \\
&=& \mathbb{E}((PX_2^{(1)})^\top PX_2^{(1)}(PX_1^{(j)})^\top PX_1^{(j)}) - \mathbb{E}((PX_2^{(1)})^\top PX_2^{(1)})\mathbb{E}((PX_1^{(j)})^\top PX_1^{(j)}) \\
&=& \sum_{i,k} \mathbb{E}(PX^{1}_{2,i} PX^{(1)}_{2,i}PX^{j}_{1,k} PX^{(j)}_{1,k}) -  \mathbb{E}(PX^{1}_{2,i} PX^{(1)}_{2,i})\mathbb{E}(PX^{j}_{1,k} PX^{(j)}_{1,k}) \\
&=&  2 \sum_{i,k} \mathbb{E}(PX^{(1)}_{2,i} PX^{(j)}_{1,k}) \mathbb{E}(PX^{(1)}_{2,i} PX^{(j)}_{2,k}) \\
&=& 2 \times \text{trace}(P\Sigma_{X_1,X_2}^{1,j}P(\Sigma_{X_1,X_2}^{1,j})^\top)
\end{eqnarray*}

Thus, not surprisingly, in the null hypothesis case, the DFA coefficients of 
each time-series are uncorrelated.

Proposition~\ref{prop:CLT} uses the delta method to derive the central limit on the $\rho_{DCCA}$ coefficients.
Now we explicitly derive the covariance matrix of $[N/n_i]\rho_{DCCA}(n_i,X_1,X_2)$ under this limit. (See~\cite{van2000asymptotic} for details on the delta method).

The input Gaussian is: 
\small
\begin{multline}
Z = ([N/n_1]F_{X_1,X_2}^2(n_1),[N/n_1]F_{X_1,X_1}^2(n_1),[N/n_1]F_{X_2,X_2}^2(n_1),\dots \\
 ,[N/n_r]F_{X_1,X_2}^2(n_r),[N/n_r]F_{X_1,X_1}^2(n_r),[N/n_r]F_{X_2,X_2}^2(n_r))
\end{multline}
\normalsize
 which has covariance $C$, say.
Then we need the covariance of $h(Z) = (\frac{Z_1}{\sqrt{Z_2 Z_3}},\dots,\frac{Z_{3r-2}}{\sqrt{Z_{3r-1} Z_{3r}}})$.
Thus since, 
\begin{eqnarray}
\frac{\partial}{\partial Z_1} h(Z)_1 &=&  \frac{1}{\sqrt{Z_2 Z_3}}  \label{eq:dis_invariance_1}\\
\frac{\partial}{\partial Z_2} h(Z)_1 &=& - \frac{Z_1}{(Z_2 Z_3)^{3/2}} \\
\frac{\partial}{\partial Z_3} h(Z)_1 &=&  -\frac{Z_1}{(Z_2 Z_3)^{3/2}} \label{eq:dis_invariance_3}
\end{eqnarray}

And since the mean under the null hypothesis of $h(Z)$ is $(0,f(H)n^{2H},f(G)n^{2G},\dots)$ (under $\mathcal{H}_0$), so that the second two derivative terms are zero, 
then we have:
\scriptsize
\begin{equation}
\text{cov}(h(Z)_{i,j}) = \\ \frac{C_{3(i-1)+1,3(j-1)+1}}{\sqrt{\mathbb{E}(Z_{3(i-1)+2})\mathbb{E}(Z_{3(i-1)+2})}\sqrt{\mathbb{E}(Z_{3(j-1)+2})\mathbb{E}(Z_{3(j-1)+2})}} \label{eq:second_order_stats}
\end{equation}
\normalsize
Thus the numerator depends only on the covariance of the DCCA coefficients\footnote{This means that the $\rho_{DCCA}$ coefficient also yields
robustness to non-Gaussianity under the null hypothesis. Intuitively the reason for this is that although $F_{X_1,X_1}^2(n_1)$ has variance of the same order of magnitude as $F_{X_1,X_2}^2(n_1)$, 
the former has non-zero mean which implies it contributes no variance asymptotically.}.

\section{Theory}
\label{sec:theory}

\begin{Prop}
\label{prop:expectation}
For any detrending degree: $\mathbb{E} [\rho_{DCCA}(n,X_1,X_2)] = 0$ if $X_1$ and $X_2$ are independent fBm; otherwise, if $X_1$ and $X_2$ are dependent components of VfBm $\mathbb{E} [\rho_{DCCA}(n,X_1,X_2)] > 0$, asymptotically for all n or $\mathbb{E} [\rho_{DCCA}(n,X_1,X_2)] < 0$, asymptotically for all $n$. More exactly:
\begin{multline}
\mathbb{E}(F^2_{1,X_1,X_2}(n)) = \\ \sigma_1 \sigma_2 \left(\rho f_1(H,G) + \eta f_2(H,G)\right)n^{H+G}(1+\mathcal{O}(1/n)) \label{eq:expectation}
\end{multline}
where $f_1(H,H) = f(H)$ from Bardet et al. \cite{DFA_asymp}, Equation 9. This implies the claim for $\rho_{DCCA}$ since the denominator of $\rho_{DCCA}$ is always positive.
\end{Prop}
\begin{proof}
Following Bardet et al. we have that $\mathbb{E}(F^2_{1,X_1,X_2}(n))  =  \text{trace}( \Sigma^{1,n}_{G,H} - \mathbb{P}_d  \Sigma^{1,n}_{G,H})$ where $\Sigma^{1,n}_{G,H}$ is the cross-covariance matrix of
the bivariate fractional Brownian motion up until time $n$.
The traces may be approximated by integrals so that we have that 
\begin{multline}
\mathbb{E}(F^2_{1,X_1,X_2}(n)) = \sigma_1 \sigma_2 n^{H+G}\bigg(\int_0^1 \mathbb{E}(X_1(s),X_2(s))^2 dt \\
		+ \int_0^1\int_0^1 p(s,t) \mathbb{E}(X_1(t),X_2(s)) ds dt \bigg)\left(1 + \mathcal{O}(1/n)\right)
\end{multline}

Where the formula $p(ns,nt)$ parametrizes the entries of the matrix $\mathbb{P}_d$ up until
$\mathcal{O}(1/n)$.
Thus it suffices to show that:
\begin{equation}
\int_0^1 \mathbb{E}(X_1(s),X_2(s))^2 dt \neq
		\int_0^1\int_0^1 p(s,t) \mathbb{E}(X_1(t),X_2(s)) ds dt 
\end{equation}
For which it is sufficient to show that  $\mathbb{E}(X_1(t),X_2(s))$
is not equal to $p(s,t)$ on the set $[0,1] \times [0,1]$ which is
true for all $G,H$ and $d$, proving the proposition.

\end{proof}

\begin{Prop}
\label{lem:cov_order}
The covariance function: \newline $\text{cov}(F^2_{1,X_1,X_2}(n),F^2_{j,X_1,X_2}(m))$ has order $j^{2G+2H-8}$ for DCCA(d) under the null hypothesis of zero correlation;
i.e.:
\footnotesize
\begin{multline}
\text{cov}(F^2_{1,X_1,X_2}(n),F^2_{j,X_1,X_2}(m)) = \\ \sigma_1^2 \sigma_2^2 g(H,G) n^{H+G}m^{H+G}j^{2H+2G-8}\left(1+ \mathcal{O}(1/\text{min}(n,m)) + \mathcal{O}(1/j)\right) \label{eq:covariance}
\end{multline}
\end{Prop}
\normalsize
\begin{proof}
Define $Q_d = I - \mathbb{P}_d$ which is the orthogonormal projection onto the complement of the polynomials of degree $d$ in $\mathbb{R}^n$.
One may show by explicit calculation\footnote{Maple sheets are available for download at \url{http://www.user.tu-berlin.de/blythed/maple_DCCA}} of the corresponding terms of the Taylor expansion of
$\text{trace}(Q_1 \Sigma_H Q_1 \Sigma_G^\top)$ (this formula is calculated in Section~\ref{sec:derivation}) that the order is correct when $d=1$; here $\Sigma_H$ refers to the univariate covariance matrix
of the terms in the first and $j^{th}$ window of size $n$.
The extension to higher order detrending ($d>1$) is straightforward.
When $H=G$, then we have:
\begin{equation}  
\text{cov}(F^2_{1,X_1,X_2}(n),F^2_{j,X_1,X_2}(m)) =  \text{trace}(Q_d \Sigma_H Q_d \Sigma_H^\top)
\end{equation}
But $\text{trace}(Q_d \Sigma_H Q_d \Sigma_H^\top) < \text{trace}(Q_1 \Sigma_H Q_1 \Sigma_H^\top)$, since $\text{trace}(Q_d \Sigma_H Q_d \Sigma_H^\top)$ is
the Froebenius norm of $Q_d \Sigma Q_d$ and $Q_d$ projects to a subspace of the orthogonal complement of the polynomials of degree up to $d$. 
Moreover using the Cauchy-Schwarz inequality for the inner-product $\text{trace}(A,B^\top)$ on matrices, the result follows for $H \neq G$.
Thus, for higher degree polynomial detrending, the order of the covariance function is less than or equal to that of DCCA(1). 
\end{proof}

\begin{Prop}
$(\frac{\sqrt{[N/n_1]}}{n_1^{H+G}} F^2_{X_1,X_2}(n_1),\dots, \frac{\sqrt{[N/n_r]}}{n_r^{H+G}}F^2_{X_1,X_2}(n_r))$  $\rightarrow \mathcal{N}(0,\Gamma(G,H))$ as $[N/n_i]\rightarrow \infty$. where $\Gamma(G,H)$ is a covariance matrix which does not
depend on $n_i$ or $N$ for large $n_i$ and $N/n_i$. \label{prop:semi}
\end{Prop}

\begin{proof}
We need the following two lemmas: \newline

\noindent \emph{Lemma 1}
$E_1^d = E_j^d$ for any $d$ (detrending degree) and $r \geq 1$. See Section~\ref{sec:derivation} for the definition of $E_j^d$. \label{lem:invar}
\begin{proof}

This may be proven directly by looking at the binomial expansion of the coefficients of the vector $(((r-1)n+1)^d,((r-1)n+2)^d,\dots,(rn)^d)$.
\end{proof}

\noindent \emph{Lemma 2} $F^2_{k,X_1,X_2}(n_i)$ is stationary for any $n_i$.

\begin{proof}

The proof that $F^2_{k,X_1,X_2}(n_i)$ is stationary is identical for the proof for DFA in Bardet et al. \cite{DFA_asymp} using our Lemma 1. No modifications are necessary for polynomial trending.
Likewise the extension to the proof for the covariance is identical.
\end{proof}

Define $Z_i = \left(X_1^{(j)}-\mathbb{P}_d\left(X_1^{(j)}\right)\right)_i$ 
and $W_i = \left(X_2^{(j)}-\mathbb{P}_d\left(X_2^{(j)}\right)\right)_i$. Then by the second lemma $(Z_1,Z_2,\dots,Z_n,W_1,\dots,W_n)$ is a stationary Gaussian vector in $j$.
We may then use conditions on functions of a Gaussian Vector sufficient for a central limit theorem (\cite{CLT_GaussianVector}, Theorem 2). It is possible to show that in the large n limit, the elements of the covariance matrix required by these conditions, which is proportional to  $(I-P)\Sigma_{1,j}(1-P)$, have order $1/j^2$, 
and thus the vector $(Z_1,\dots,Z_n,W_1,\dots,W_n)$ satisfies the conditions\footnotemark[5] (since if the order is slower than this then the order of the DFA coefficients $\text{cov}(F_{1,X_1,X_1}(n),F_{j,X_1,X_1}(n))$ is slower than $j^{-4}$ which is impossible by Proposition~\ref{lem:cov_order}). The generalization to the multivariate case is straightforward by way of the Cramer Wold device \cite{cramer1936some}.  
\end{proof}

\begin{Prop}
\label{prop:CLT}
Under the null hypothesis of zero correlation for two channels of the multivariate fractional Brownian motion, with Hurst parameters $H,G \in (0,1)$, the $\rho_{DCCA}$ coefficients $(\sqrt{[N/n_1]}\rho_{DCCA}(n_1), ... , \sqrt{[N/n_r]}\rho_{DCCA}(n_r))$ obey a multivariate central limit theorem as $n_1,\dots,n_r \rightarrow \infty$ and $[N/n_i] \rightarrow \infty$ to a limit which is independent of $n_i$ and $N$ and depends only on $H$ and $G$.
\end{Prop}

\begin{proof}
This then follows from the multivariate delta method levering the results of the central limit theorem for DFA \cite{DFA_asymp} and DCCA (our Proposition~\ref{prop:semi}). 
The limiting covariance matrix is given by Equation~\eqref{eq:second_order_stats}.
\end{proof}

\begin{Prop}
\label{prop:worst_case}
The worst case covariance calculated in Algorithm~\ref{alg:wccov} correctly upper bounds the probability of a type I error.
\end{Prop}

\begin{proof}
Since all $\rho_{DCCA}$ coefficients are positively correlated with each other (trace of a positive semidefinite matrix), then by choosing the maximum correlation and variances at each position in the matrix guarantees the upper bound.
\end{proof}

\begin{Prop}
\label{prop:Taqqu}
The low frequencies of $X_j(t)$ tend to those of a fBm; more formally: as $n \rightarrow \infty$ then we have that:
 $L_j(n)^{-1/2} X_j(nt) \xrightarrow{d} \sigma B_H(nt)$ 
with $\sigma >0$.
\end{Prop}

\begin{proof}
See \cite{taqqu1975weak} for the proof.
\end{proof}

\begin{Prop}
\label{prop:non_gauss}
Assuming independence of $Y_1$ and $Y_2$ and assume that the $\text{Pr}(\text{sup}_{s\in[0,t]}X_j(s) \geq x) \leq \frac{C t^H}{x}$ as is the 
case for fractional Brownian motion (\cite{vardar2009results} (Theorem 2.2)) and assume that the H{\"o}lder exponents of $X_1$ and $X_2$ are greater than or equal to the exponents of fractional Brownian motion, then as $n \rightarrow \infty$ and $[N/n] \rightarrow \infty$ then $\frac{1}{n^{H+G}}F^2_{X_1,X_2}(n) \xrightarrow{d} \frac{\sqrt{L_1(n) L_2(n)}}{n^{H+G}}F^2_{X_1^B,X_2^B}(n)$, where
the $X_j^B$ are fractional Brownian motions with the same Hurst parameters as $Y_1$ and $Y_2$.  \end{Prop}
\begin{proof}

Out proof consists of three steps; firstly we show that the DCCA coefficients of the process given by subsampling $Y_1$ and $Y_2$ tend to those of fractional Brownian motion in the limit of subsampling, by using
Proposition~\ref{prop:Taqqu}.
Secondly we show that for a given ratio of subsampling to window size the DCCA coefficients of both the unsubsampled version and the subsampled version tend to each 
other. Thirdly we show that these imply that in the limit of window size, the DCCA coefficients tend to the law of the coefficients of fractional Brownian motion. 
Together these imply that in a certain limit, the DCCA coefficients obey a central limit theorem with the same limiting distribution as the central limit 
proven above for fractional Brownian motion.

Part I: 
We define the subsampled processes as $L_j(n)^{-1/2} X_j(nt)$.
It is easy to show that the DCCA coefficients of the subsampled process
 tend to those of fBm; this follows by Proposition~\ref{prop:Taqqu} for large window sizes.
 
%

Part II: We have:
\small
\begin{equation}
F_{1,X_1,X_2}^2(n) = \frac{1}{n^2}\sum_{t=1}^n (X_1(t)- \mathbb{P}_d(X_1(t))) (X_2(t)- \mathbb{P}_d(X_2(t)))
\end{equation}
\normalsize
Let $\widehat{k}(X_j) = (X_j(k),X_j(2k),\dots,X_j(t))$ and let $f^1_n$ be the least squares poly. fit of degree $d$ to 
$X_1$ and $f^1_{n/k}$ to $\widehat{k}(X_j)$. Now we can use the Cauchy Schwarz inequality:
\scriptsize
\begin{eqnarray}
&& |F^2_{1,X_1,X_2}(n) - F^2_{1,\widehat{k}(X_1),\widehat{k}(X_2)}(n/k)|^2  \\
&=& | <X_1-f^1_n,X_2-f^2_n> \\ 
&-& <\widehat{k}(X_1)-f^1_{n/k},\widehat{k}(X_2)-f^2_{n/k} >|^2 \\
&=& | <X_1-f^1_n,X_2-f^2_n> - <\widehat{k}(X_1)-f^1_{n/k},X_2-f^2_{n} > \\
&+&  <\widehat{k}(X_1)-f^1_{n/k},X_2-f^2_{n} > \\
&-& <\widehat{k}(X_1)-f^1_{n/k},\widehat{k}(X_2)-f^2_{n/k} >|^2 \\
& \leq&  |<X_1-f^1_n-\widehat{k}(X_1)+f^1_{n/k}, X_2-f^2_n>|^2\\
 &+& |<X_2-f^2_n-\widehat{k}(X_2)+f^2_{n/k}, \widehat{k}(X_1)-f^1_{n/k}>|^2  \\
 &\leq&  |X_1-f^1_n-\widehat{k}(X_1)+f^1_{n/k}| \times | X_2-f^2_n| \\
 &+& |X_2-f^2_n-\widehat{k}(X_2)+f^2_{n/k}| \times | \widehat{k}(X_1)-f^1_{n/k}| \\
 &\leq& (|X_1-\widehat{k}(X_1)| + |f^1_n-f^1_{n/k}|) \times | X_2-f^2_n| \\
  &+& (|X_2-\widehat{k}(X_2)| + |f^2_n-f^2_{n/k}|) \times | \widehat{k}(X_1)-f^1_{n/k}|  \label{eq:C_S}
 \end{eqnarray}
\normalsize

The most problematic term which we need to bound here is the second term in each sum under the 
bracket.

It is possible to show, using Riemann sums, that, if $f^*$ is the poly. of degree $d$ which minimizes the generative
mean squared error to $X$ ($f^* = \text{min}_f \int_0^1 (f(nt)-X(nt))^2 $) then:

\begin{eqnarray*}
|X - f_n|^2 &=& \frac{1}{n}\sum_{i=1}^n (X(i) - f_n)^2 \\
&=& |X-f^*|^2 + \mathcal{O}(\text{max}((X-f^*)^2)/n^2)
\end{eqnarray*}

Thus since least squares is a convex and differentiable optimization
problem, we have that $f_n$ and $f^*$ are close to one another and
so too are $f_n$ and $f_{n/k}$.
I.e.:

\begin{eqnarray*}
|f_{n/k} - f_n|^2 &\leq& |f_{n/k} - f^*|^2  +  |f_{n} - f^*|^2 \\
&\leq& r_{n}q_k
\end{eqnarray*}

To obtain these rates $r_n$ and $q_k$ we need to obtain the parameter of strong convexity 
for the m.s.e. ( $\int_0^1 (f(nt)-X(nt))^2 dt $): the generative Hessian of the m.s.e. for the linear polynomials is proportional to:

\begin{eqnarray}
\partial_a \partial_a \int_0^1 (ant+b-X(nt))^2dt &=& 1/3n^2\\
\partial_a \partial_b \int_0^1  (ant+b-X(nt))^2dt&=& n \\
\partial_b \partial_a \int_0^1  (ant+b-X(nt))^2dt&=& n \\
\partial_b \partial_b \int_0^1  (ant+b-X(nt))^2dt&=& 1 \\
\end{eqnarray}

This implies that the smallest eigenvalue $\lambda_1$ of the Hessian matrix is a constant and greater than zero.
Similarly for the polynomials of degree $d$ we have a non-degenerate and symmetric Hessian whose
smallest eigenvalue is thus greater than zero.

Therefore \cite{nesterov2004introductory}: 
\begin{eqnarray*}
&& \int_0^1 (X_1(nt)-f^1_n(nt))^2 dt - \int_0^1(X_1(nt)-f^*(nt))^2 dt \\
&\geq& \lambda_1 (a_n-a^*)^2 +  (b_n-b^*)^2
\end{eqnarray*}

Here we have the difference of two mean squared errors on the interval $[0,n]$.
By assumption we have that with probability $1-\frac{C n^H}{n}$:
\small
\begin{eqnarray*}
\frac{1}{n}\sum_{i}^n ((a_n-a^*) i - (b_n-b^*))^2 
&\leq& \frac{n-1}{2}((a_n-a^*)^2 +  (b_n-b^*)^2) \\
&\leq& D\frac{n-1}{2}(\text{max}((X-f^*)^2)/n^2) \\
&\leq&  D' n
\end{eqnarray*}
\normalsize

Thus with probability $1-\frac{C (n/k)^H}{n/k}$:
\small
\begin{eqnarray*}
\frac{1}{n/k}\sum_{i}^{n/k} ((a_{n/k}-a^*) i - (b_{n/k}-b^*))^2 
&\leq&  D' {n/k}
\end{eqnarray*}
\normalsize

Moreover we have that $X_1$ and $X_2$ are almost surely H\"older continuous with exponent $H-\phi$ for any $\phi<0$ \cite{decreusefond1999stochastic}
Thus there exists $E$ s.t.:

\normalsize
\begin{eqnarray*}
|X_1-\widehat{k}(X_1)|^2 &=& \frac{1}{n} \sum (X_1(i) - X_i(mod(i,k)+k))^2 \\
&\leq& E k^{2H-2\phi}
\end{eqnarray*}
and with probability $1-\frac{C n^H}{n}$:
\begin{eqnarray*}
|X_1-f^1_n| &\leq& \text{sup}|X_1| \\
&\leq& n
\end{eqnarray*}
Thus we have, with probability at least $1-\text{max}\left(\frac{Cn^H}{n} \frac{C (n/k)^H}{C{n/k}},\frac{C n^G}{n} \frac{C(n/k)^G}{{n/k}}\right)$ and a constant $A$:
\small
 \begin{eqnarray*}
  && \frac{1}{n^{2H+2G}}|F^2_{j,X_1,X_2}(n) - F^2_{j,\widehat{k}(X_1),\widehat{k}(X_2)}(n/k)|^2 \\ 
  &\leq& A \frac{n (n/k) k^{H-\phi} +n  (n/k) k^{G-\phi}}{n^{2H+2G}} 
\end{eqnarray*}
\normalsize
in probability.

Thus let $0<\delta < 1$ then, with probability at least $1-C' \text{max}\left( n^{(H-1)(1-\delta)}, n^{G-1)(1-\delta)}\right)$, letting $k = n^\delta$:

 \begin{eqnarray*}
  && \frac{1}{n^{2H+2G}}|F^2_{j,X_1,X_2}(n) - F^2_{j,\widehat{k}(X_1),\widehat{k}(X_2)}(n/k)|^2 \\ 
  &\leq& A \frac{ n^2 (n^\delta)^{H-\phi} + n^2 (n^\delta)^{G-\phi}}{n^{2H+2G+\delta}} \\
 \end{eqnarray*}

And this final expression tends to zero regardless of which values $H$ and $G$ take on.
This implies that $F^2_{j,X_1,X_2}(n)$ tends in probability to the $k$-subsampled version as $n,k \rightarrow \infty$.
Since $F^2_{j,X_1,X_2}$ is a stationary sequence, this convergence holds for all $j$.
Thus we can use Boole's equality to show that $F^2_{X_1,X_2} \rightarrow F^2_{\widehat{k}(X_1),\widehat{k}(X_2)}$ with constraints on $N$.

I.e.:
  \begin{eqnarray}
    \text{Pr}\bigg( \frac{1}{n^{2H+2G}}|F^2_{X_1,X_2}(n) - F^2_{\widehat{k}(X_1),\widehat{k}(X_2)}(n/k)|^2 \\ 
  \leq BA \frac{ n^2 (n^\delta)^{H-\phi} + n^2 (n^\delta)^{G-\phi}}{ n^{2H+2G+\delta}} \bigg)   \label{eq:converge_multiple}\\
  \geq 1-[N/n]\frac{2}{\pi}\text{max}\left( n^{(H-1)(1-\delta)}, n^{(G-1)(1-\delta)}\right)
 \end{eqnarray}
 
Thus: provided $N = o(\text{min}(n^{(1-H)(1-\delta)+1},n^{(1-G)(1-\delta)+1}))$ then the convergence in probability holds.

PART III: to show convergence in distribution we use the portmanteau theorem \cite{klenke2008probability}. Let $s$ be a bounded continuous function
and define $S(X_1,X_2) = \frac{1}{n^{H+G}}F^2_{X_1,X_2}(n)$ and $S_{X_1^B,X_2^B}(n)  = \frac{L(n)}{n^{H+G}}F^2_{X_1^B,X_2^B}(n)$ and $k = n^\delta$.
 \begin{multline}
   |\mathbb{E}(s(S_{X_1,X_2}(n)) - \mathbb{E}(s(S_{X_1^B,X_2^B}(n)))| \\
   \leq |\mathbb{E}(s(S_{X_1,X_2}(n)) -  \mathbb{E}(s(S_{\widehat{k}(X_1),\widehat{k}(X_2)}(n/k)))| \\ + | \mathbb{E}(s(S_{\widehat{k}(X_1),\widehat{k}(X_2)}(n/k)))-\mathbb{E}(s(S_{X_1^B,X_2^B}(n/k)))| \\
   +  | \mathbb{E}(s(S_{X_1^B,X_2^B}(n/k)))-\mathbb{E}(s(S_{X_1^B,X_2^B}(n)))|
 \end{multline}

The first and the last lines converge in virtue of Part II (conv. in prob. implies conv. in dist.) and the second line converges in virtue of Part I.
 \end{proof}

\begin{Prop}
\label{prop:DFA_non_gauss}
Under the assumptions of Proposition~\ref{prop:non_gauss}, as $n \rightarrow \infty$ and $[N/n] \rightarrow \infty$ then $\frac{1}{n^{H+G}}F^2_{X_j,X_j}(n) \xrightarrow{d} \frac{L(n)}{n^{H+G}}F^2_{X_j^B,X_j^B}(n)$, where
the $X_j^B$ is a fractional Brownian motion with the same Hurst parameter as $Y_i$.
\end{Prop}
\begin{proof}
The proof is identical to the proof made for the DCCA coefficients.
\end{proof}

%
%
%
%

\begin{Prop}
Under the assumptions of Proposition~\ref{prop:non_gauss}, as $n \rightarrow \infty$ and $[N/n] \rightarrow \infty$ then we have convergence to the central limit of Proposition~\ref{prop:CLT} for
the DCCA correlation coefficients assuming independence of $Y_1$ and $Y_2$.
\label{prop:glue}
\end{Prop}

\begin{proof}

Equation~\eqref{eq:converge_multiple} implies that we require:

\begin{multline}
 \sqrt{[N/n]}\left(\frac{ n^2 (n^\delta)^{H-\phi} + n^2 (n^\delta)^{G-\phi}}{ n^{2H+2G+\delta}}\right) \rightarrow 0 \\
\text{i.e.  } \sqrt{N}\frac{ n^2 (n^\delta)^{H-\phi} + n^2 (n^\delta)^{G-\phi}}{ n^{2H+2G+\delta+1/2}} \rightarrow 0 \\
\end{multline}

This statement is satisfiable while keeping $[N/n] \rightarrow \infty$. Thus we choose the rate on $N$ and $n$ to coincide 
with this rate and the rate given by Proposition~\ref{prop:non_gauss}.
\end{proof}

Note here that the rate we have calculated may be considerably improved by replacing the bound given by Boole's inequality with a more 
sophisticated bound calculated using distribution specific information w.r.t the structure of our time-series.

We finally note a mathematically interesting observation made by numerical evaluation of Equation~\eqref{eq:covariance_fun} of the Appendix.

\begin{Conj}
The correlation between DCCA correlation coefficients is maximized asymptotically for low $H$ and the variance maximized for high $H$ under the null hypothesis.
\end{Conj}
%
%
%

\begin{figure}[h]
\begin{center}
$\begin{array}{c c c}
\includegraphics[width=50mm]{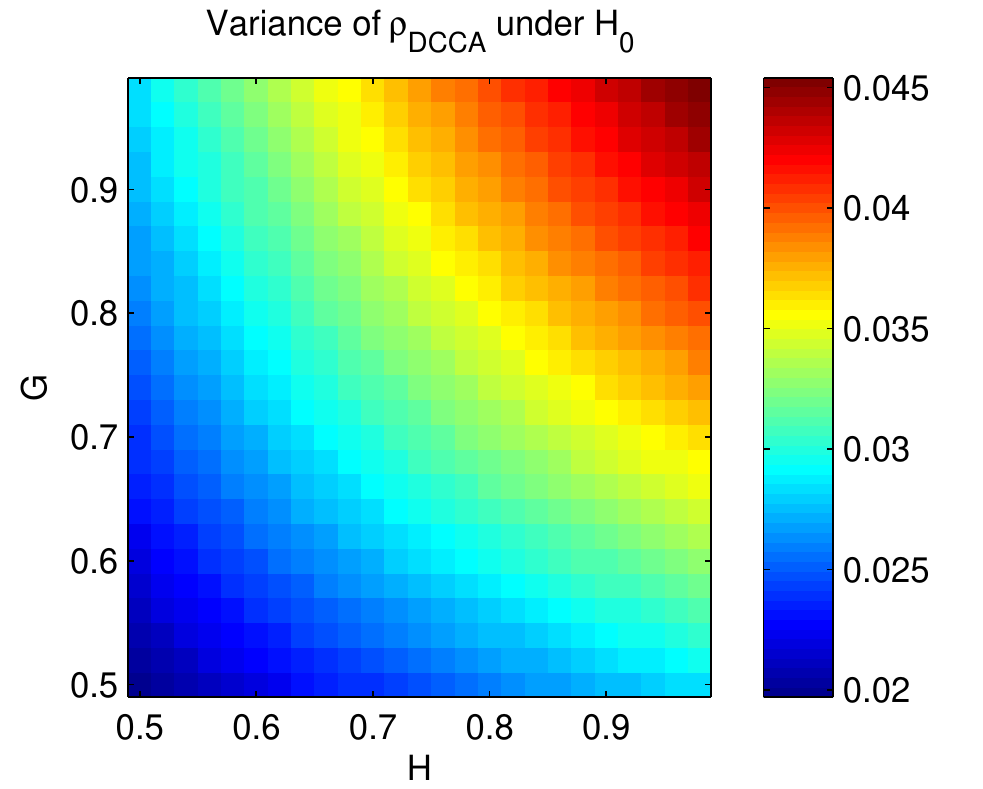} & \includegraphics[width=50mm]{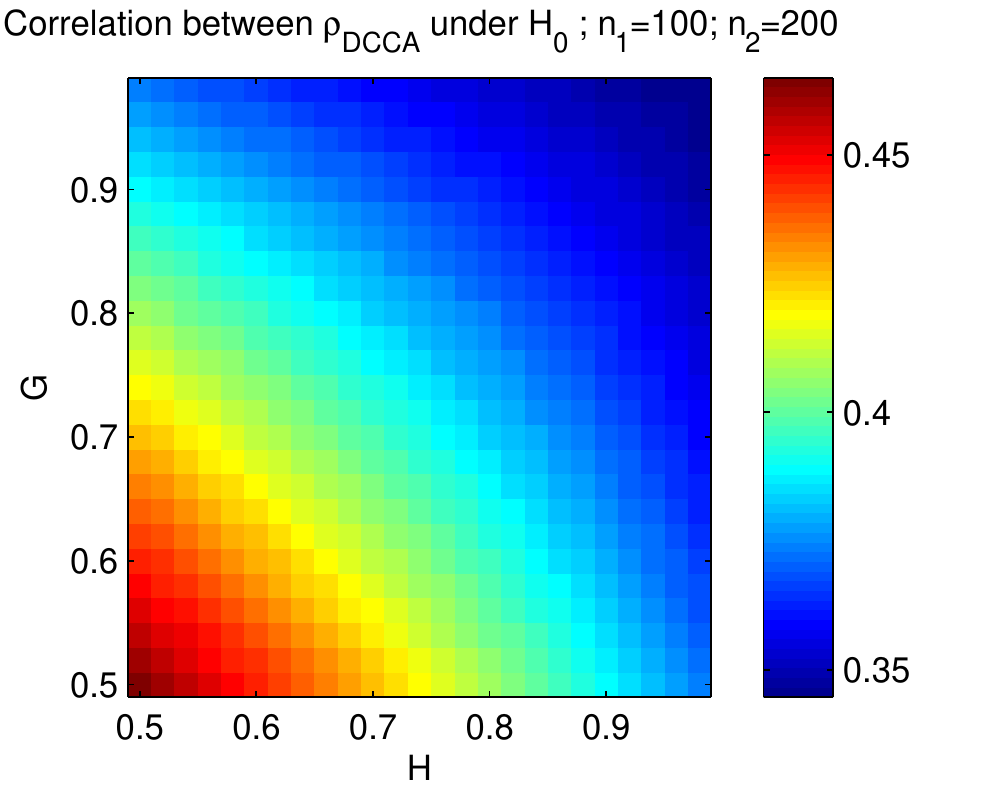}  &  \includegraphics[width=50mm]{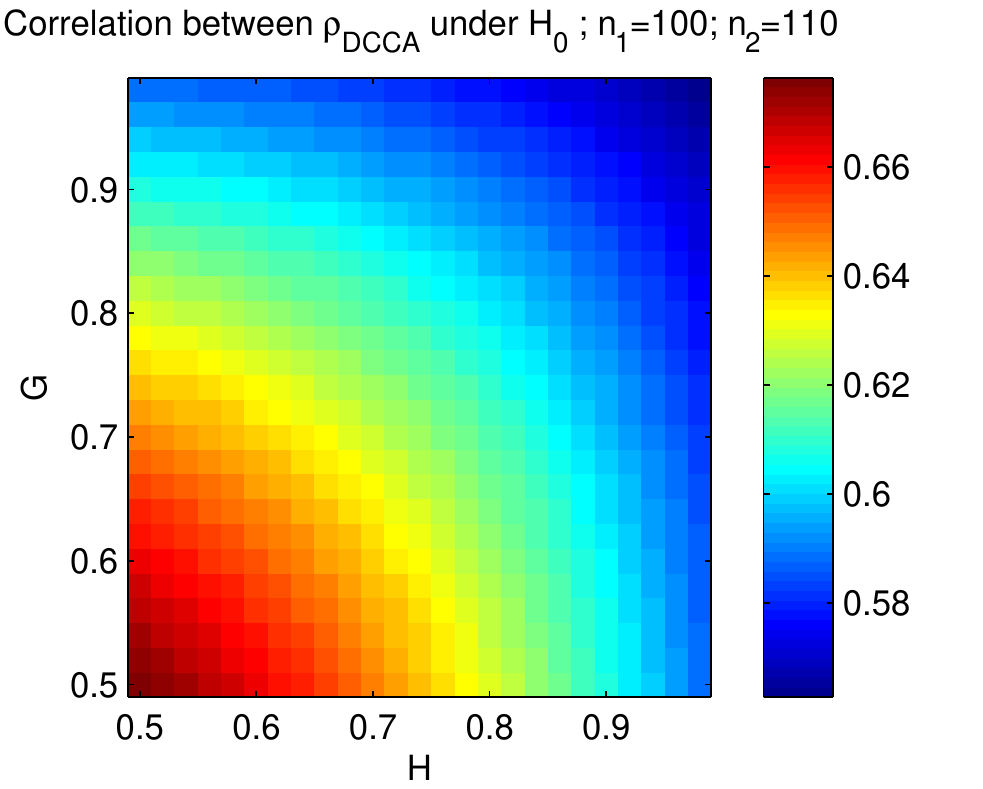} \end{array}$
\end{center}
\caption{The figure demonstrates the conjecture that variance in the $\rho_{DCCA}$ is maximized for large $H$ and $G$ but that the correlation between coefficients in minimized for small $H,G$. \label{fig:maximum_corr_against_H_G}}
\end{figure}

\section{Software}
\label{sec:software}
In the accompanying software to this paper, we supply the following functions programmed in \newline MATLAB\footnote{The software is available for download from: \url{http://www.user.tu-berlin.de/blythed/DCCA_matlab}}:
\begin{enumerate}
\item{\tt DCCA\_fast.m} -- a fast implementation of DCCA for arbitrary detrending degree $d$.
\item{\tt DCCA\_rho.m} -- calculates the DCCA correlation coefficients $\rho_{DCCA}$ on the basis of {\tt DCCA\_fast.m}.
\item {\tt covariance\_rho\_DCCA.m} -- exactly calculates $\text{cov}(\rho_{DCCA}(n),\rho_{DCCA}(m))$ for small $N$.
\item {\tt asymptotic\_test.m}  -- implements the test given by Algorithm~\ref{alg:test} using tabulation of the 
	covariance plus the central limit theorem, using sampling to calculate the quantiles of the normal distribution.
	Tabulation of the covariance function for $d=1$ is saved in {\tt .mat} files in the download package.
	\end{enumerate}

\end{document}